\documentclass[10pt]{article}

\usepackage{fullpage}
\usepackage{amsmath,color,graphicx}
\usepackage{multirow} 
\usepackage{amsthm}
\usepackage{amsmath}
\usepackage{hyperref}
\usepackage{tikz}
\usepackage{sidecap}
\usetikzlibrary{arrows,shapes,trees}
\usetikzlibrary{backgrounds}
\usepackage{graphicx}
\usepackage{caption}
\usepackage{bbm}
\usepackage{multirow}
\usepackage{subfigure}
\usetikzlibrary{arrows,positioning}
\usetikzlibrary{decorations.markings}
\usepackage{enumerate}
\usepackage{amssymb}
\usepackage[utf8]{inputenc}
\usepackage{marvosym}
\usepackage{authblk}
\newtheorem{theorem}{Theorem}
\newtheorem{example}{Example}
\newtheorem{remark}{Remark}
\newtheorem{definition}{Definition}
\usepackage{algorithm, algorithmic}

\usepackage{natbib}
 \bibpunct[, ]{(}{)}{,}{a}{}{,}%

\title{Chain of Antichains: An Efficient and Secure Distributed Ledger Technology and Its Applications}

\author{Jinwook Lee \thanks{jl3539@drexel.edu}}

\affil{\small{Decision Sciences,  Drexel University, Philadelphia, PA 19104, United States}}

\author{Paul Moon Sub Choi \thanks{paul.choi@ewha.ac.kr}}

\affil{\small{Ewha School of Business, Seodaemun-gu, Seoul 03760, South Korea}}

\date{}

\begin{document}

\maketitle

\begin{abstract}

Since the inception of blockchain and Bitcoin (\cite{bitcoin}), a decentralized-distributed ledger system and its associated cryptocurrency, respectively, the world has witnessed a slew of newer adaptations and applications. Although the original distributed ledger technology (DLT) of blockchain is deemed secure and decentralized, the confirmation of transactions is inefficient by design. Recently adopted, directed acyclic graph (DAG)-based distributed ledgers validate transactions efficiently without the physically and environmentally costly building process of blocks (\cite{dagcoin}). However, centrally-controlled confirmation against the odds of multiple validation disqualifies the DAG as a decentralized-distributed ledger. In this regard, we introduce an innovative DLT by reconstructing a chain of antichains based on a given DAG-pool of transactions. Each antichain (box) contains distinct nodes whose approved transactions are recursively validated by subsequently augmenting nodes. The boxer node closes the box and keeps the hash of all transactions confirmed by the box-genesis node. Designation of boxers and box-geneses is conditionally randomized for decentralization. The boxes are serially concatenated with recursive confirmation (boxchain) without incurring the cost of box generation. Rewards (boxcoin) are paid to the contributing nodes of the ecosystem whose trust is built on the doubly-secure protocol of confirmation. A value-preserving medium of payment (boxdollar) is among numerous practical applications discussed herein.\\


\end{abstract}

\section{Introduction and our motivation}

The original vision of the Internet is to construct a global system of interconnected computer networks. The Internet has revolutionized the modern society and prosaic activities -- we became constantly connected. 
Yet, even though we are in the era of shared economies with help of the Internet, the individual service providers and their central authorities are not as much mutually-beneficial as they can be. For example, the drivers of shared rides and the hosts of shared accommodation can be better-compensated with more decentralized discretion. The central authorities of those industries can streamline the ever-increasing costs of cyber-security by decentralizing their centralized systems with the advent of decentralized-distributed networks. Decentralization clearly synergizes all stakeholders of the emerging value chains of shared economies, regardless of the size and side of transactions. The conventional centralized mechanism can be decentralized using the distributed ledger technology (DLT), more colloquially referred as the blockchain technology.

Upon the creation of blockchain by \cite{bitcoin}, the DLT has gained dramatic attention over the last decade among developers, practitioners, and academics alike. The unprecedentedly multiplied market values of cryptocurrencies, including blockchain-associated Bitcoin, have inspired a bevy of practitioners and researchers to implement a variety of real-life applications on distributed peer-to-peer (P2P) network systems. The original blockchain is considered decentralized and secure with recursively serial confirmation. However, the environmentally and physically arduous process of block building (``mining'') is the root cause of systemic inefficiency. This causes ramification (``forking'') of cryptocurrencies beyond Bitcoin. Recently, directed acyclic graph (DAG)-based distributed ledgers -- including Hashgraph (\cite{baird}), IOTA (\cite{tangle}), Byteball (\cite{byteball}), etc. -- have emerged as a new generation of DLTs. A DAG bears the shape of an irregularly tangled leaf with recursively validating, binomial sub-trees, and it is computationally efficient without mining. However, a DAG-based DLT typically requires centrally-managed nodes to perform the final confirmation of transactions and, to the worse, suffers from the odds of double spending.

Our approach to DLT is synthetic: Combine the decentralized and secure features of blockchain with the efficiency of DAGs. An antichain (``box'') is a set of distinct, unrelated elements (``nodes'') that orthogonally complements a chain. As a DAG evolves in the original ledger, we synchronously reconstruct a dual ledger with a chain of antichains (``boxchain'') as a serial concatenation of boxes of nodes that originate from the underlying DAG. Because a node in the DAG-ledger recursively approves up to two transactions of two previous nodes, a reflecting box in the dual ledger, in effect, recursively confirms the transactions of the last box. 
This doubly-secure consensus protocol is as follows: First, a new node in a box immediately verifies and validates the approved transactions of the previous node in the box. Second, when the last node in the box (``boxer,'' a light node\footnote{A lightweight node (or light client) only references the full node's copy with less required memory capacity and processing power.}) is determined the head node of the box (``box-genesis,'' a full node\footnote{Since the mining process (i.e., solving a partial hash inversion problem) is not required in this system, prohibitive computing equipment is not necessary for the full node, which only requires sufficient data storage capacity.}) is randomly chosen among good-standing nodes other than the boxer. Third, the box-genesis finally confirms all validated transactions within the box and keeps the hash of confirmed transactions, followed by timestamp synchronization with all the previous box-geneses. This means that the box-geneses must have a copy of the entire network. Such timestamp synchronization also takes place to boxers, which only keeps  hash pointers.
At the core of trust, or governance, on the boxchain ecosystem, this consensus protocol is numerically impossible for a malignant intention of delayed or multiply-validated transactions. While the original DAG-ledger and the dual ledger synchronize and expand, there is no mining game involved contributing to the efficiency of the system. In sum, our dual ledger-keeping algorithm is decentralized, efficient and secure.

Although currently available P2P payment systems appear more facile and expeditious than before, they can improve significantly with embracement of DLT as both share similarities in chains of digital signatures of asset transfers and P2P networking, etc. In this sense, e-commerce marketplaces can be an ideal place to implement the idea and application of DLT. However, a cryptocurrency with volatile market value is not suitable for usage in payment, and in order for it to serve as a medium of exchange or a store of value, its value has to be stable or pegged against a fiat currency. In these regards, our DLT provides two distinct cryptocurrencies: boxdollar, a value-preserving medium of payment or store of value (``stablecoin''); and boxcoin, a crypto asset. Such a dual-currency system is desirable in sustaining a decentralized-distributed ledger network with proper incentives as well as expanding the ecosystem for diverse usages of stablecoins. In case of international payments for merchandise purchases, first and foremost is the highest reliability and stability -- a proper digital currency must be able to function like a fiat money as a useful medium of exchange and a dependable store of value. Thus, pegging boxdollar to a fiat money (e.g., the U.S. dollar) is straightforward and intuitive. In order to maintain the reliability and stability of boxdollar, the optimized rebalancing of currency holdings per various criteria, e.g., \cite{kataoka} ``safety first model'' or conditional value-at-risk model, can be considered. Besides, boxdollar can be used as a local currency in a municipality with stable tax revenues. As the local governments with local stablecoins mutually agree to render them compatible, these value-preserving cryptocurrencies may eventually act as quasi-fiat currencies.\footnote{See Section \ref{bxd} for the discussions and models of boxdollar.}

While we need a stablecoin like boxdollar to finance purchases and various transactions, we also view the distributed ledger system as a society of healthy extent of ``greed.'' It is not difficult to see that concerns may be overdone if a society only consists of good-standing, law-abiding participants. However, tantamount rewards paid out in boxcoin, a crypto asset, are essential for the peers with good-standing track records. In order to encourage rewarding behaviors, a fair opportunity will be given to every peer to play a vital role in the system. In our ecosystem, there are two particular positions -- the boxer and the box-genesis -- to which any qualified participants can be designated multiple times. Besides, small fees will be charged to issued transactions to provide more security to the system. Without fees, there can be nonsensical transactions burdening the efficiency and security of the system. The amount of fees will be from negligible to small, depending on the matter.\footnote{See Section \ref{incentivesystem} for more details of the incentive system.}  The mentioned elements such as rewards and fees must be realized instantaneously upon occurrence of events -- giving prompt incentive and motivation to the peers. When the ecosystem gets larger, more robust and multiplied, it is obvious that the value of its currency will gain given a fixed quantity of currency in circulation. This possibility of capital appreciation is a great incentive for the peers to abide by the code of conduct. Given these accounts, we are motivated to build a new distributed peer-to-peer system with two distinct cryptocurrencies: boxdollar and boxcoin. Boxdollar is a stablecoin for our daily use while boxcoin is a crypto asset to improve and manage our ecosystem along with all participants' active engagement.

The organization of this research is as follows: Section \ref{box} introduces our new idea of the dual ledger-keeping algorithm, which is based on decomposition of partially ordered sets (DAGs). Based on the primal layer of a DAG-based ledger, a dual-ledger of boxchain (chain of antichains) is synchronized in the dual layer. A comparison with existing DLTs and applications is provided. In Section \ref{stoc} we present the mathematics of our DLT: the nonhomogeneous Poisson process, combination of multiple transaction flows, and the compound Poisson distribution. Section \ref{bxd} provides the models of boxdollar and discusses its potentials in the value chain, followed by our concluding remarks in Section \ref{bond}.

\section{The dual ledger-keeping algorithm for constructing an effective and secure distributed ledger}\label{box}
\subsection{DAG: distributed system, differences from blockchain, and previous works}\label{2.1}

A DAG, or acyclic digraph, consists of two entities: vertices and edges, i.e., $\mathcal{G} = (\mathcal{V},\mathcal{E})$, which consists of a set $\mathcal{V}$ of vertices (or nodes) and a set $\mathcal{E}$ of edges (or directed arcs).
An edge is an ordered pair $(i, j)$ meaning outgoing from vertex $i$, incoming to vertex $j$, a pair $(j, i)$ means from vertex $j$ to vertex $i$.
A DAG is deemed an generalized blockchain because the graph itself is a ledger for storing transactions where each vertex is intinsically a single block. 
Note that in this section we use the term DAG and $\mathcal{G}$ interchangeably. Besides, the terms ``vertex'' and ``node'' are equivalent -- ``vertex'' will be used when we are more focused on graph itself and ``node'' when transaction processes and their related topics are discussed. Some of the notable differences of a DAG from the original blockchain ideas are as follows:
\begin{itemize}
\item The main structure of a blockchain is a single chain that consists of blocks (a chain of blocks) and each block is a set of multiple transactions.
If we look at each transaction as a single block, it is not difficult to understand that such blocks and their connections can be depicted as a DAG. 
A DAG can also be expressed in terms of trees -- note that if there are $N$ vertices (or nodes), then there can be {\it at least} $N$ binary trees (see Section \ref{boxblock}).

\item Blockchain is based on synchronous time stamps, while a DAG is built asynchronously operating. However, our dual-ledger system is updated in a nearly synchronous manner (see Section \ref{tip2}.).

\item In a DAG there is none to force to separate participants into different categories. On the other hand, Bitcoin and some other cryptocurrencies justify to maintain two separate types of participants, required either to issue or to approve transactions. There are no ``miners'' creating blocks and receiving rewards in DAG.
\end{itemize}

Let us briefly mention how a DAG system works as a distributed ledger. Transactions are issued by nodes and their edges show what previous transactions they approved - users need to approve previous transactions in order to issue their own transaction. All transactions have its own weight, and in our case they are all equal to one which is adaptable as the system evolves. Speaking of the approval and the weights, the validation comes with cumulative addition of weights and it goes to all the connected previous nodes (i.e., ancestors of the node). 
 Unapproved nodes in a DAG are called ``tips.'' See Figure \ref{fig0111}.
 
 A transaction validation process is simple -- A node selects some tips at random and verify their validity in order to issue its own transaction. It is assumed that the approving nodes checks whether or not the two connected transactions are conflicting in a diligent and honest manner. Note that this validation is different from the final confirmation in terms of consensus mechanism. 
 A node even before the first transaction is called the ``genesis.'' The genesis node has all tokens created in the beginning of the DAG (no more tokens will be created) and then send the tokens out to several founder nodes. In order to issue a transaction a node needs to approve the two tips and solve a cryptographic puzzle (i.e., a hash inversion problem) similar to those in the Bitcoin blockchain. For more technical details we refer the reader to \cite{byteball},  \cite{iota}, \cite{dagcoin}, \cite{dag}, etc.  \cite{byteball} and \cite{tangle}  investigate many important deterministic and stochastic aspects of the network flows on DAGs, respectively.. We will discuss both deterministic and stochastic aspects here for improvements of the framework.

\begin{figure}[h!]
\begin{center}
\includegraphics[width=0.99\textwidth]{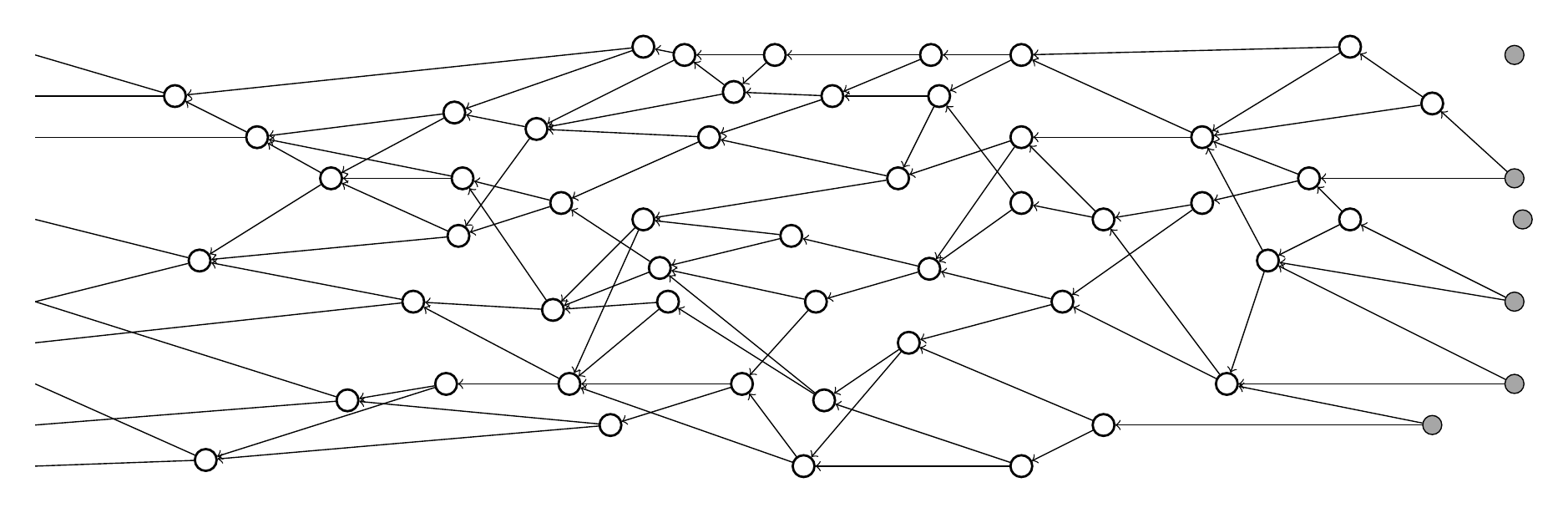}
\end{center}
\caption{Transactions on a DAG. Unapproved, gray nodes are tips.}\label{fig0111}
         \end{figure}

\subsection{DAG, structure of trees and related set representation}\label{boxblock}

We continue to develop and modify the existing cryptocurrency models on the DAG, and more importantly, how we make the connection between the blockchain and a DAG. We will begin with an introduction of basic notions and develop such concepts using the system of distinct representatives, i.e., decomposition of a partially ordered set (poset). A simultaneous creation of an efficient network layer (called the dual-layer) will follow, and this is the core of our dual-ledger system.

Our space is a DAG $\mathcal{G}$, where vertices (or nodes) designate transactions and edges (or directed arcs) denote how they are connected -- the edge $(i,j)$ means that node $i$ approves node $j$. In our network system a node selects exactly ``two'' tips at random and verify their validity in order to issue its own transaction. More details are presented after this section. There is another notation we need to define together with $\mathcal{G}$. $\mathcal{N}$ designates the network associated with $\mathcal{G},$ which consists of sub-networks $\mathcal{N} = \{ \mathcal{N}_1, \dots, \mathcal{N}_V \},$ where $V = |  \mathcal{V} |.$ Thus, each network contains a subset of $\mathcal{G}$ where a large number of vertices can be included. However complicated these networks look, there are always starting and ending points, and these are the latest vertex (the source node) and the genesis (the sink node), respectively.

A DAG is a combination of trees which consist of child nodes and their parent and ancestor nodes. 
For a tree, {\it Parent} means the predecessor of a node; {\it Child} is any successor of a node; {\it Siblings} are a pair of nodes that have the same parent; {\it Ancestor} is the set of predecessor; {\it Descendant} is the set of successors; {\it Subtree} denotes a node with its descendants. 

\begin{remark}[Binary Trees]
As a node is supposed to approve two previous transactions, it is a binary tree. It is easy to observe that the whole system can be separated into multiple binary trees. (There at least $V$ binary trees if the total number of vertices is denoted by $V$.)
\end{remark}

In Section \ref{2.1} we briefly mentioned the basic mechanics of issuing and approving transactions in a DAG. 
In connection with trees, if a node $v$ approves two previous transactions then such nodes that issued the approved transactions become parent-nodes (or parents) of node $v$. Note that this validation of the parents automatically approves the parents of that parents, followed by such recursion all the way to the genesis. If we select a single parent node out of two parent nodes then it forms a chain, i.e, a linearly linked, totally ordered set. This is very efficient and is regarded as one of the main benefits to manage transactions in a DAG.

However, there are a typically large number of such chains, which are only partially ordered as there are incomparable ancestors in terms of subset relation,  or validation relation. A DAG is a poset itself. The  multiple chains of a DAG is the main reason why the system is asynchronous, which makes it almost impossible for all peers to agree to a single version of the truth. This means the final confirmation from a reasonable consensus mechanism would be unreachable if the network system is solely based in a DAG. This motivates our dual ledger-keeping algorithm.
 
\begin{remark}
A DAG is a poset by subset relation. 
\end{remark}

Now let us turn our attention to see what happens to a DAG when a node is approved in terms of the cumulative weight.

\begin{remark}[Cumulative weight on a node, a counting measure of integrity and trust]
As mentioned earlier, we assume all nodes have their own weight equal to one. In addition to the references to the parents (by hashes), the cumulative weight is a good measure for the integrity and trust.
By the validation of a single child node, its all ancestors' weights will be added exactly by one, which is convenient as a counting measure. 
In other words, the weight of a node provides the system with information of the number of child nodes at the moment. The exact cumulative weight of a node is not known to its neighbors due to the asynchronous nature of a DAG-based network. 
\end{remark}

\begin{remark}[Inclusion of transaction information on nodes]
Note that a child node encompasses its parent node by referencing its parent's hash. 
All information of ancestors (including the genesis) can be obtained (reference does not take any memory, only point to the location) in a recursive manner, which means that a child node has more information than its parent nodes. 
\end{remark}

Speaking of the information inclusion by nodes let us use the following notations: $$T_k = \{ \text{transaction information of node $i_k$}\}.$$
 Suppose that node $i_7$ approved node $i_6$ and node $i_5$, and node $i_6$ approved node $i_4$ and $i_3$, and node $i_5$ approved nodes $i_2$ and $i_1$ as in Figure \ref{fig01111}. Then node $i_7$ can get the information of all 6 previous nodes and itself (of course), node $i_6$ includes nodes $i_4$, $i_3$ and itself, and node $i_5$ includes nodes $i_2$, $i_1$ and itself.
For $k=1, \dots, 7$, let the set $A_k$ designate the union of all ancestor nodes from the point of view of node $i_k$.
This example can be written up as:
$$A_7 = T_1 \cup T_2 \cup \dots \cup  T_7,$$

$$A_6 = T_6 \cup T_4 \cup T_3, \ \ A_5 = T_5 \cup T_2 \cup T_1$$
and
$$A_1 = T_1, A_2 = T_2, A_3 =  T_3, A_4 = T_4.$$
Thus, we have 
 $$A_7 =  T_7 \cup A_6 \cup A_5,$$ followed by the subset relation
 \begin{equation}\label{7sets}
A_7 \supseteq A_6, A_7 \supseteq A_5, A_6 \supseteq A_4, A_6 \supseteq A_3, A_5 \supseteq A_2,  A_5 \supseteq A_1,
\end{equation}
$A_5 \nsupseteq A_6$ and $A_5 \nsubseteq A_6$ as they are only partially ordered. This is a simplistic example, but illustrates some important ordering on components of the network. Refer to Figure \ref{bitree}, where child and parent nodes are upside down. In a tree graph, it is typical to have child nodes below, but in our special validation operation child nodes approve the transactions of their parent nodes. Again, the direction of edges are not from the genesis, but towards it. 

\begin{figure}[h!]
\begin{center}
\includegraphics[width=0.99\textwidth]{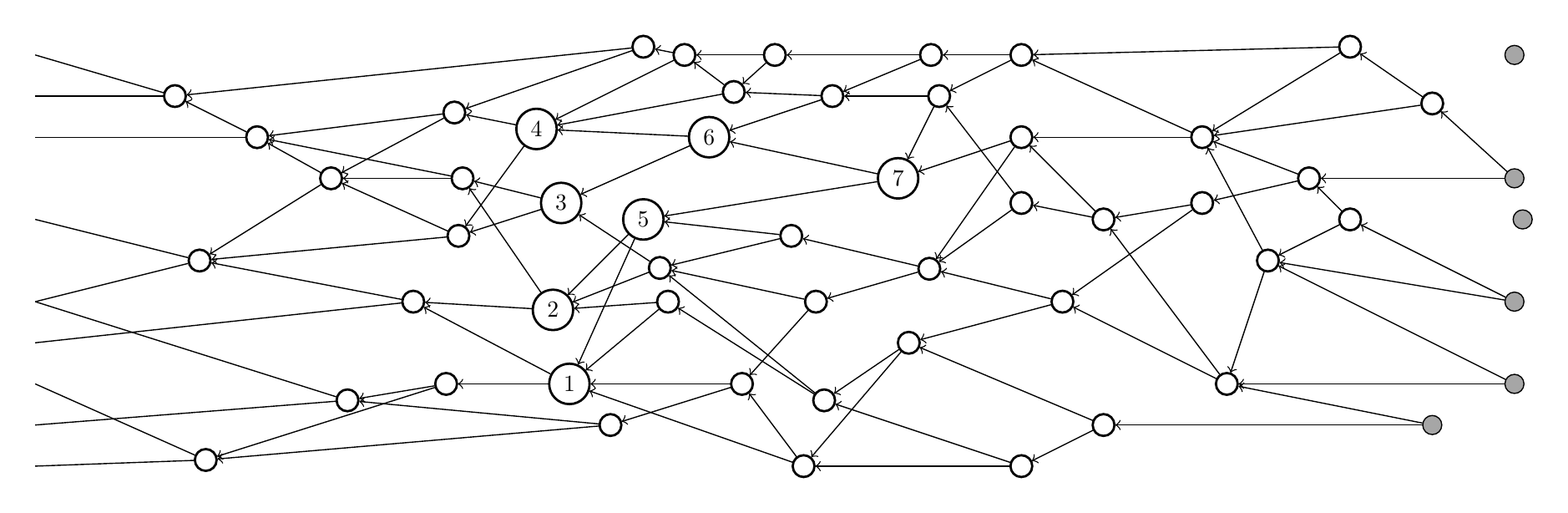}
\end{center}
\caption{Transactions on a DAG. Unappoved, gray nodes are tips. (Compare it with Figure \ref{bitree}.)}\label{fig01111}
         \end{figure}

\begin{figure}[ht!]{\centering
         \subfigure[Hasse diagram on partial order by subset relation.]{\includegraphics[width=0.35\textwidth]{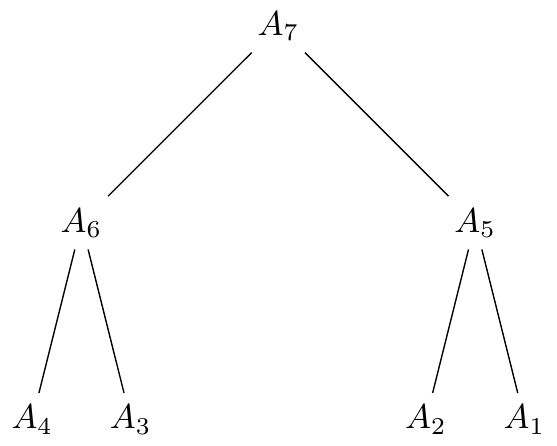}}\qquad
         \subfigure[A set of binary trees with referencing hash functions.]{\includegraphics[width=0.35\textwidth]{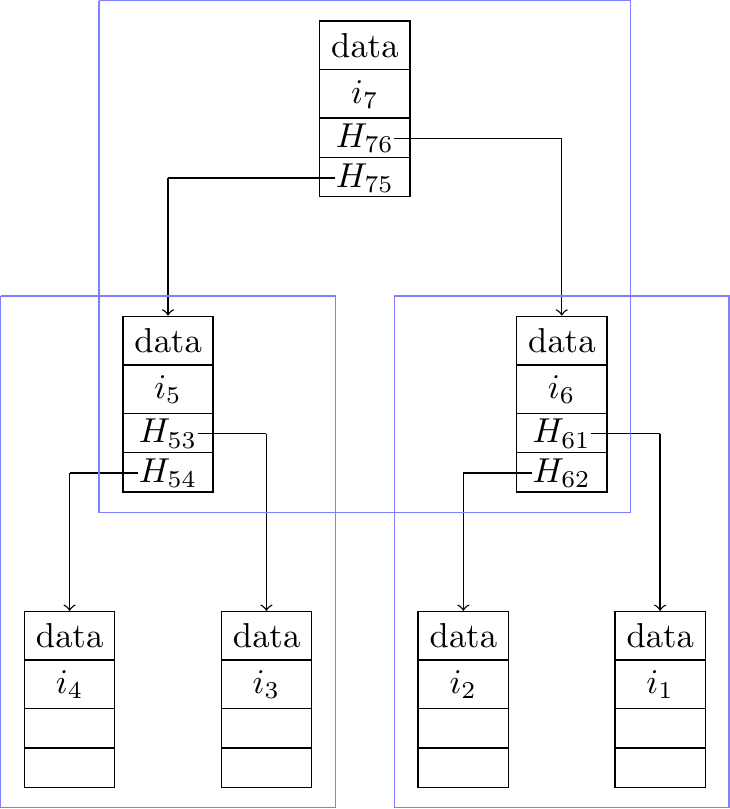}}
       	  \caption{Inclusion of information in a subset of a DAG of Figure \ref{fig01111}}  \label{bitree} }
         \end{figure}

\begin{remark}[Not a typical binary tree structure]
Again, note that the hierarchical relationship in a typical tree has ancestors above their predecessors. Our case is the opposite due to the selection-approval operation by a child node. This means, the root of a tree is a common child (the youngest ``single'' one!) of all the others. A child is a ``superset'' of the parent nodes.
\end{remark}

Desides, the tree in Figure \ref{bitree} can be divided into three binary trees (i.e., repeated structural forms as illustrated), which will help solve a variety of problems in a recursive or iterative manner. We employ such recursion for our consensus protocol (in Section \ref{consensusprotocol}).

\subsection{Primal and dual spaces,  and the dual ledger-keeping system}\label{det}
\subsubsection{Partially ordered set and its related concepts}
In this section we present models and related algorithms  making the system more integral, secure and efficient.
We provide some basic definitions and results in connection with posets. 

We say that two elements $x$ and $y$ of $S$ are {\it comparable} if $x \leq y$ or $y \leq x$, otherwise $x$ and $y$ are {\it incomparable}.
If all pairs of elements are comparable then $S$ is {\it totally ordered} with respect to $\leq$.
An element $M$ of $S$ is called a {\it maximal} element in $S$ if there exist no $x \in S$ such that $M \leq x$ (i.e., $M \leq x \rightarrow M =x$). 
An element $m$ of $S$ is called a {\it minimal} element in $S$ if there exist no $y \in S$ such that $y \leq m$ (i.e., $y \leq m \rightarrow m = y$).
An element $G$ of $S$ is called a {\it greatest} element in $S$ if $x \leq G$ for all elements $x$ of $S$, and the term {\it least} element is defined dually.
A {\it chain} (or totally ordered set or linearly ordered set) in a poset $S$ is a subset $C \subseteq S$ such that any two elements in $C$ are comparable (A chain is a sequence).
An {\it antichain} in a poset $S$ is a subset $A \subseteq S$ such that ``no'' two elements in $A$ are comparable. 

\begin{remark}[Blockchain is serial, thus totally ordered, while a DAG is only partially ordered.] 
The main structure of a blockchain is a series of blocks -- it is a serial concatenation of blocks, hence a totally ordered set. A DAG is only partially ordered and a mixture of piece-wise serial networks, seemingly disorganized. 
A DAG can also be restructured in a meaningful way by decomposing a poset into chains and antichains.
\end{remark}

It is known that if every chain of a poset $S$ has an upper bound in $S$, then $S$ contains at least one maximal element (Zorn's lemma).
If $x, y \in S$, then we say that $y$ {\it covers} $x$ or $x$ is covered by $y$, denoted $x \lessdot y$ or $y  \gtrdot x$, if $x < y$ and no element $u \in S$ satisfies $x < u < y.$
The chain $C$ of $S$ is maximal if it is not contained in a larger chain of $S$.

\begin{remark}\label{cover}
It is desirable that all maximal chains have the same length from the genesis to the corresponding tips. Typically some of the tips are not in the same antichain. 
\end{remark}

The size of the largest antichain is known as the poset's width; The size of the longest chain is known as the poset's height. This will be used in (\ref{nm}).
\begin{remark}
The height is from the genesis to a tip in the most recent antichain. 
\end{remark}

A poset can be partitioned using chains (see, e.g.,  \cite{poset}, \cite{poset1}, \cite{poset2}, etc.).
A poset can also be partitioned using antichains which is closely related to our formula construction as shown in \cite{poset3}. 
The following theorem is considered self-evident and is presented without the proof.
\begin{theorem}[Dual of Dilworth] \label{partition} 
Suppose that the largest chain in a poset $S$ has length $n$. Then $S$ can be partitioned into $n$ antichains. 
\end{theorem}

\begin{definition}\label{rankfunction} A rank function of a poset $S$ is function $r: S \rightarrow \{0\} \cup \mathbb{N}$ having the following properties:
\begin{itemize}
\item [(i)] if $s$ is minimal, then $r(s) = 0.$
\item [(ii)] if $t$ covers $s$ (i.e., $t \gtrdot s$), then $r(t) = r(s) + 1.$ 
\end{itemize}
\end{definition}

\begin{remark} 
In a DAG, the minimal element is the genesis. 
\end{remark}

\begin{definition}[\cite{union1}]\label{def10} On a finite poset $S$ with length $n$, ordered by inclusion, the reverse rank function $\rho:S \rightarrow \{1, \dots, n\}$ is defined  by
$$\rho(E) = \sum_{i}   \mathbbm{1}_{E \subseteq M_i},$$
where $E$ is any element of $S$ and $M_i$'s are incomparable maximal elements of $S$. The reverse rank function $\rho(E)$, a counting measure, returns the number of maximal elements containing $E$.  
\end{definition}

\begin{remark}
In a DAG, tips are the maximal elements.
\end{remark}

\begin{theorem}
Given a finite poset $(S, \subseteq)$, $\max_{E \in S} \rho(E)$ is the width (i.e, the size of the largest anti-chain) of a poset and equals to 
$$\min \left\{m \ \big| \text{ maximal chains }C_1, \dots, C_m \text { with } S = \bigcup_{i=1}^m C_i \right\}.$$  
\end{theorem}
\begin{proof}
Obvious by \cite{poset} and the reverse rank function in Definition \ref{def10}. 
\end{proof}

\begin{definition}[system of distinct representatives] Suppose that $A_1, A_2, \dots, A_N$ are sets. The family of sets $A_1, A_2, \dots, A_N$ has a system of distinct representatives if and only if there exist distinct elements $z^{(1)}, z^{(2)}, \dots, z^{(N)}$ such that $z^{(i)} \in A_i$ for each $i = 1, \dots, N$.
\end{definition}\label{sdr}

\begin{theorem}[Duality]
The minimum number of non-redundant edges from an antichain to its previous one is equal to the maximum number of distinct representatives in the latter antichain.
\end{theorem}

\subsubsection{A dual layer construction and two essential roles}

Based on the aforementioned notions, we present the main algorithm which constructs a new layer on top of a convoluted, partially ordered DAG.
We will decompose a DAG using chains and  antichains. The original DAG is preserved as it normally operates, and there is a dual-layer being constructed almost synchronously based on subset relation among the nodes on the DAG.

\begin{figure}[ht!]
\begin{center}
\includegraphics[width=0.65\textwidth]{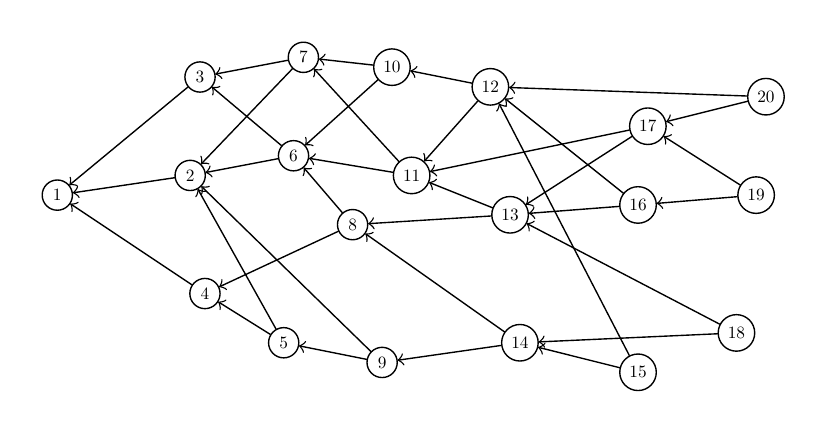}
\end{center}
\caption{Transaction units connected in a DAG whose genesis is node 1.}\label{fig000}
         \end{figure}

Let us consider a DAG of 20 nodes (with four tips of nodes 15, 18, 19, 20) as in Figure \ref{fig000}.
The $i$-th-created node shows its node number $i$, i.e, node $1$ is the genesis node. Child nodes are the supersets of their parent nodes. Incomparable nodes form an  antichain and the set of incomparable child nodes must be a ``cover'' of the set of incomparable parent nodes. In Figure \ref{fig001}, for example, the nodes $8, 9, 10, 11$ are included in an antichain (the third bluebox from the left) since they are mutually incomparable.
The ``union'' of those four nodes (i.e., antichain) include the ``union'' of nodes 5, 6, 7 (i.e, antichain) -- no single node among the four nodes $8, 9, 10, 11$ include all three nodes 5,6,7 as a node can only approve two previous nodes.

As described in Figure \ref{fig001}, by the decomposition using antichains, any entangled DAG can be restructured into a linked list (i.e., a chain of blue boxes in Figure \ref{fig001}), resembling a blockchain. For this reason,  the blockchain is considered as a special version of a DAG, which is  a generalized blockchain.

\begin{remark}
We do not change a given DAG to a single chain; Rather, we preserve the original DAG-based ledger to make the system more secure and trusted. This is the embarkation of our dual ledger-keeping algorithm. \end{remark}

Specifically, we present a series of descriptions as in Figures \ref{fig000}, \ref{fig001}, \ref{fig0010} and \ref{fig0_0}. 
Let $B_k$, $k=1, \dots, 6$ denote the bluebox in Figure \ref{fig001} from the left to right in order. $B_1$ has elements of nodes 2,3,4, $B_2$ has elements of nodes 5,6,7, and so on.
Then $B_1 \subseteq B_2 \subseteq B_3 \subseteq B_4 \subseteq B_5 \subseteq B_6$ and  $B_1\lessdot B_2 \lessdot B_3 \lessdot B_4 \lessdot B_5 \lessdot B_6$ as a bluebox covers its most recent and previous one.

\begin{figure}[ht!]
\begin{center}
\includegraphics[width=0.70\textwidth]{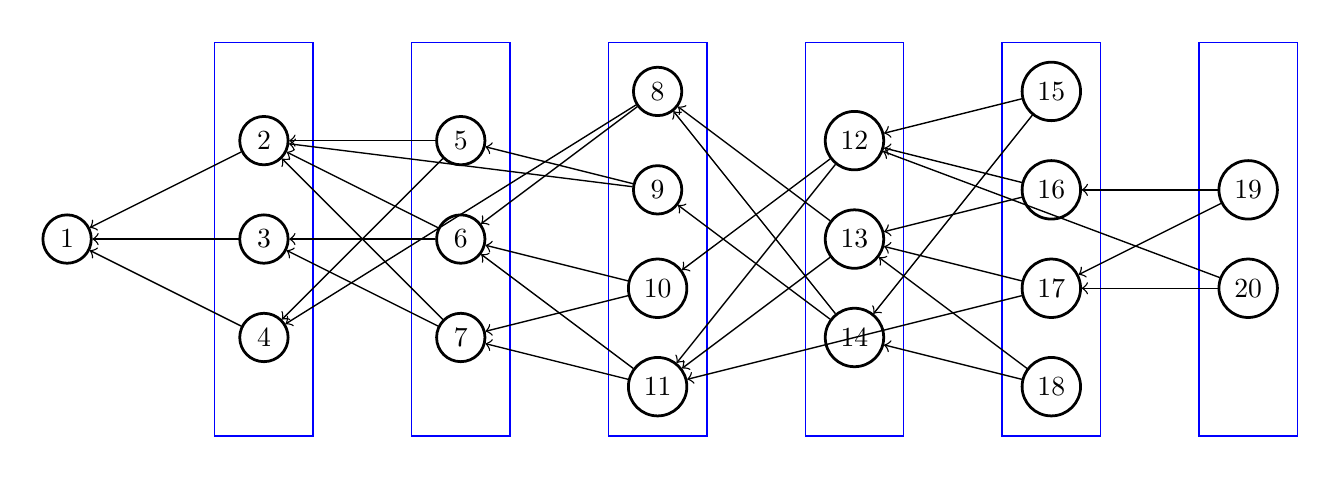}
\end{center}
\caption{A series of boxes (antichains) with incomparable units, constructed based on the partial order by their subset relations.}\label{fig001}
         \end{figure}

\begin{figure}[ht!]
\begin{center}
\includegraphics[width=0.70\textwidth]{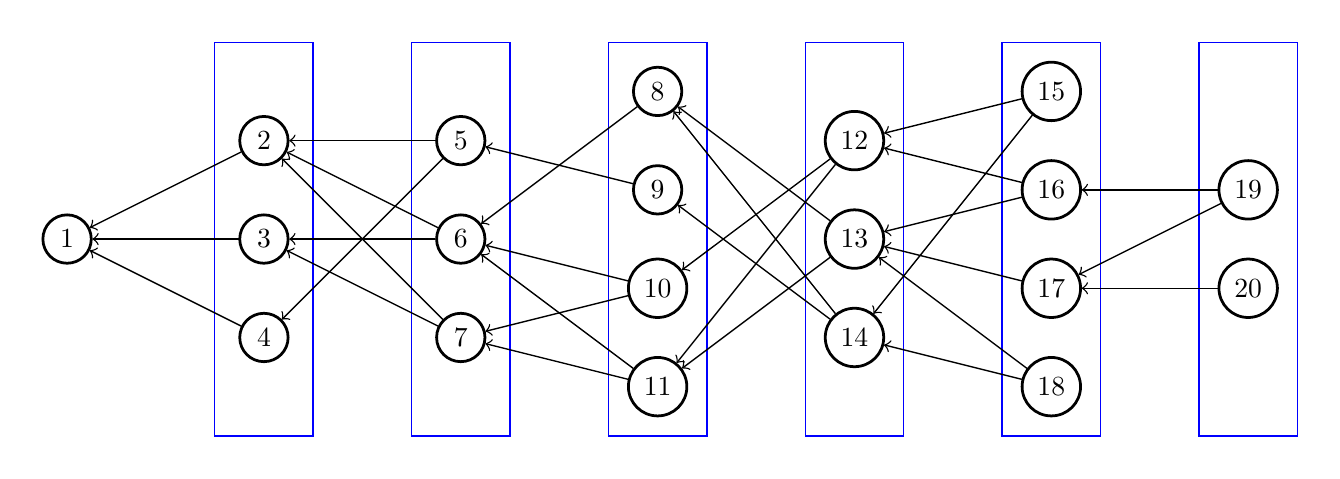}
\end{center}
\caption{A series of boxes (antichains) after redundant edges (approvals) are removed.}\label{fig0010}
         \end{figure}

In including the most recent previous transaction, if the same address issues multiple transactions (multiple nodes) then it is desirable to form a single linked list (i.e., a chain or a totally ordered set). Tracking its own series of transactions will clarify such task \footnote{A rapid succession of the multiple transactions using different addresses is a typical way to attack the system. More details are discussed in Section \ref{consensusprotocol}.}.

Figures \ref{fig001} and \ref{fig0010} are slightly different: A few number of edges in Figure \ref{fig001} are missing in Figure \ref{fig0010}.
One of the two edges are removed from nodes 8, 9, 17, 20 as some of them are redundant for the sake of the consistency of the system.
Note that node 9 approved nodes 5 and 2, but node 5 approved node 2 earlier. Node 9's approval on node 5 means that it automatically approved node 2 as node 5 is the child of node 2. 
Node 9 is a child of node 5, so node 9 cannot be a child of node 2 since node 2 is a parent of node 5 -- it will hurt the parent-child relationship. (The nodes 17, 13, 11 are of the same case.).

The approval selections of nodes 8 and 20 are valid but undesirable. 
They did not approve redundant transactions and are, thus, harmless in this example.  
However, in a real situation some idling nodes might select unnecessary and old transactions that are already approved by many other nodes. This may cause a higher complexity in a variety of important operations. Further, this may put nonsensical and illogical cumulative weights on the DAG because it interferes with efficient ordering operation of the system. Selection of nonsensical transactions does not take place in our dual-ledger system, which will be discussed in  Section \ref{consensusprotocol}.

The blue boxes in Figure \ref{fig0010} are the antichains as well as the SDRs (if we see distinct chains from the nodes). From the geometric perspective, the elements of the most recent antichain can be considered as vertices of upper bounded orthants in a multidimensional space, and it follows that the earlier antichains' elements are inside of such orthants.

Note that there is no link among the elements of the same antichain. Let us pick a single node and appoint it to represent and act for its siblings. The nodes we will choose are the ones that joined their antichains the latest. 
In this example, such nodes are 4, 7, 11, 14, 18, 20 and they are all over the place in a DAG (see Figure \ref{fig0_0}.). We call them the box-closing-members, or ``boxers."

\begin{remark}[The roles of the box-closing-members, or boxers]\label{boxersrole}
The roles of the boxers are adaptable as the system evolves. 
One of the crucial roles is that they communicate with their siblings, the members of the same antichain, and broadcast to their siblings if there are some notable changes on the system. Moreover, there is a boxers' own network, on which  they can efficiently communicate with each other. The network is just a single chain connected by edges among them, but it is a semi-hidden network since it's exclusive for the boxers.
\end{remark}

\begin{figure}[ht!]
\begin{center}
\includegraphics[width=0.70\textwidth]{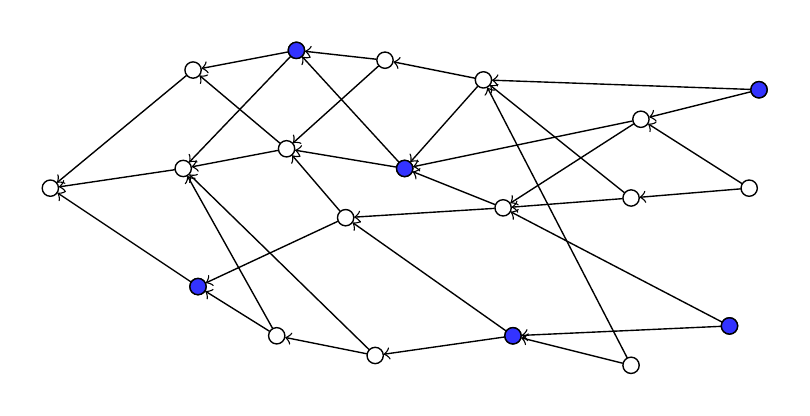}
\end{center}
\caption{Primal layer: Transaction units in a DAG, our primal space with boxers in blue}\label{fig0_0}
         \end{figure}

\begin{figure}[ht!]
\begin{center}
\includegraphics[width=0.70\textwidth]{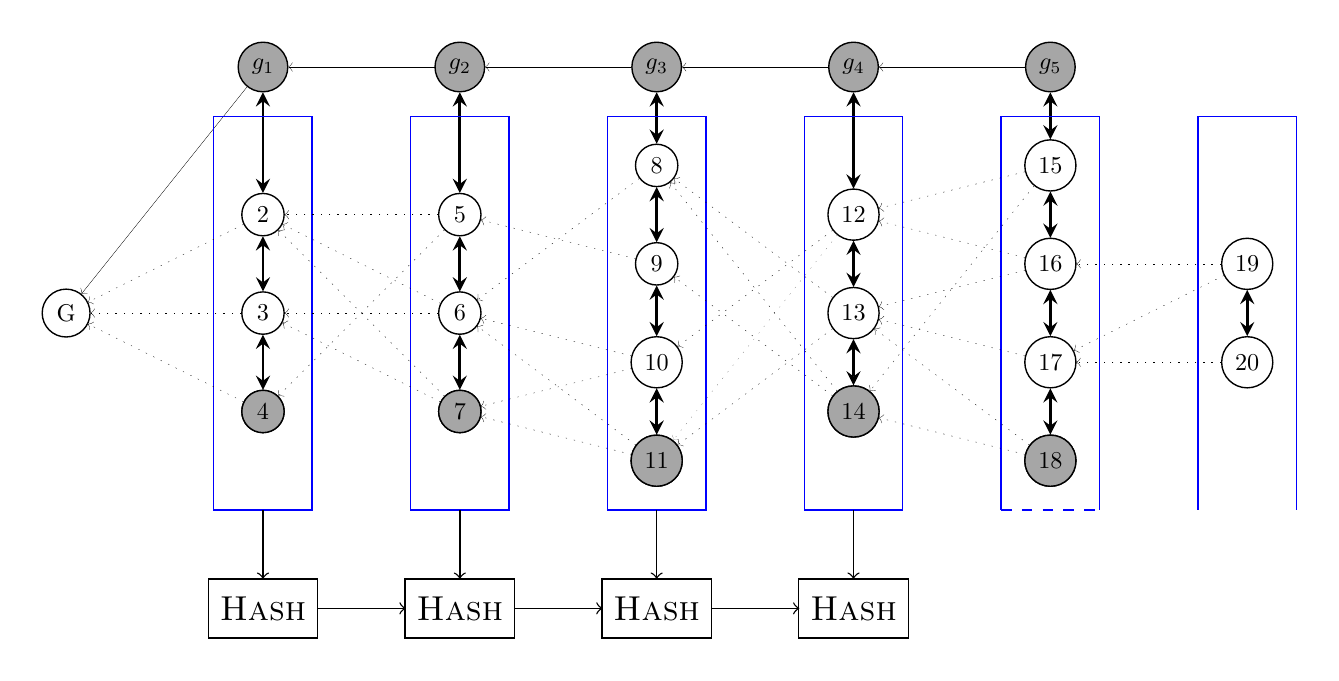}
\end{center}
\caption{Dual layer: Boxers (in gray) and their siblings in the dual layer, constructed in real time.}\label{fig00110}
         \end{figure}

The boxes are the antichains and they can be concatenated to form a chain: the chain of antichains. 
Individual box-genesis will be denoted by $i$-genesis (or $g_i$) if it's inside the $i$-th box (antichain). 
The $i$-genesis is the child of ($i-1$)-genesis and the parent of $i+1$-genesis, and ($i+1$)-genesis is the child of $i$-genesis and the parent of ($i+2$)-genesis, and so on. 
\begin{equation}\label{genesis}
\text{genesis} \lessdot \cdots \lessdot  (i-1)\text{-genesis} \lessdot    i\text{-genesis} \lessdot    (i+1)\text{-genesis} \lessdot \cdots
\end{equation}
There is also a chain (as in the above covering inequalities in (\ref{genesis})) of box-geneses which is the same height as the whole DAG -- The height of a DAG is the size of its longest chain.

\begin{remark}[The roles of the box-genesis]\label{genesisrole}
The box-genesis (a full-node) is selected among the nodes of good-standing track records. Once selected, one is responsible to check most recently approved transactions by its siblings, followed by announcing the final confirmation of their validity. Such transactions are in the parent-box (most recent previous box). When the final confirmation is placed, the box will finally be closed in agreement with a boxer and no more validations can be placed. More details are presented in Section \ref{consensusprotocol}.
\end{remark}

If the system has a fixed upper bound of the width or size of antichain, say $M$, and if the total number of vertices of a given DAG is $N$, then
\begin{equation}\label{nm}
\text{the height of the whole DAG is at most }\lceil N/M \rceil,
\end{equation}
which is a markedly important measure owing to the fact that the system governs the quantity $M$ per Monte Carlo simulation. Therefore, the quantity of (\ref{nm}) can only be estimated in the primal space. From the dual space the height of the whole DAG can easily be found, which equals to the number of closed antichains plus two at most.
 We will use this measure for further discussions, together with another criterion -- the time constraint $\tau$ for each antichain. This is also a vital constraint to control the system. As both $M$ and $\tau$ and their roles are presented in Remark \ref{dualc}, they are both paramount criteria for a boxer selection. The size of antichain will be determined by the box-closing boxer.

This chain of box-geneses is a network from the latest $i$-genesis (a root node) to the genesis (a sink node and the origin), hidden to users,  which helps make the system more secure and compact on the dual layer (in Figure \ref{fig00110}). The regular nodes are not informed of such chain unless the need arises. The detailed roles of the box-geneses are presented in Section \ref{consensusprotocol}.

Simply put, the box-making algorithm is as follows: open the box, examine and put items in the right packaging boxes, then close the boxes.
Using the boxes based on the algorithm detailed below, we construct a new layer which is a combination of a semi-hidden network and a normal chain.
This new layer is a separate space but it corresponds to the original DAG. 
This layer will be updated in real time as transactions are released and validated on the DAG . 

We present two algorithms -- Algorithm 1 at a fixed point of time and Algorithm 2 in real time. 
The following is to show how to fill a box at a specific time point with unassigned nodes. 
Let $B_i$ denote the $i$-th antichain (box) and $V_i'$ the corresponding set of vertices which approved two previous transactions but not yet validated by others. Put another way, $V_i'$ is the set of tips, about to issue new, or unvalidated, transactions. 

\begin{algorithm}
\caption{Constructing the $i$-th box ($B_i$) at a point of time for $i=1, \dots, n$}
\begin{algorithmic} 
\STATE Suppose the box $B_{i-1}$ is full, and $V_i' = \{v_{i_1}, \dots, v_{i_n} \}$.
\WHILE{there exists $v \notin B_i$ for every $v \in V'$ with $B_i \cup \{v\}$}
\IF{$v$ is a cover of the element(s) of $B_{i-1}$}
\STATE $B_i \leftarrow B_i \cup \{v\}$
\STATE $V_i'  \leftarrow V_i'  \backslash v$
\ELSE
\STATE $V_{i+1}'  \leftarrow V_{i+1}' \cup \{v\}$
\ENDIF
\ENDWHILE \\
Stop; The box $B_i$ has been filled up, let the last $v$ be the boxer. Then the $i$-th genesis $g_i$ is assigned, and we move on to the next box $B_{i+1}$ with $V_{i+1}' $
\end{algorithmic}
\end{algorithm}

\begin{remark}
The box-genesis will be selected among good-standing nodes only after the boxer is determined, which is an imperative rule to keep the system secure. (see Remarks \ref{genesisrandom} and \ref{genesisrandom1}.)
\end{remark}

\begin{remark}[Timestamp server]
Timestamps form a chain, and each timestamp includes the has of previous timestamps as in Figure \ref{fig00110}.
\end{remark}

\begin{remark}
Each box is identified by a hash, and is linked to its previous box by referencing the previous box's hash.
\end{remark}

The following is how to build a dual-layer in real time. See Remark \ref{dualc} for the constraints $M$ and $\tau$ for the size of antichains. 
A node that approves two previous transactions will be assigned to the corresponding antichain (based on the partial order).

\begin{algorithm}
\caption{Constructing the antichains in real time (i.e, a single node at a time)}
\begin{algorithmic} 
\STATE Suppose $|B_i| = k_i$ for all $i$.
\IF{a node $v$ just approved two previous transactions, and is a cover of the element(s) of $B_{j-1}$}
\STATE $B_j \leftarrow B_j \cup \{v\}$
\STATE Update $B_j = \{ v_{j}^{(1)}, \dots, v_{j}^{(k_j)}, v_{j}^{(k_j+1)} \}$
\ENDIF\\
\end{algorithmic}
\end{algorithm}

\begin{remark}[Boxchain on the dual-layer]
The dual layer corresponding to the original DAG-based layer has a chain of antichains (boxchain), which resembles a blockchain since each box consists of transactions and such boxes are serially linked.
\end{remark}

Similar to Bitcoin transactions, in our dual-ledger system most of the conflicts are easily determined if they are in different antichains (boxes).
If they are in different antichains (i.e, totally ordered), the latter one will be rejected.   
If suspicious transactions are in the same antichain, their cumulative weights can be compared. 
If such transactions happen to have the same weight (same number of descendants), some tie-break rules can be devised, including a fixed maximum width (the size of the largest antichain), a boxers role, etc. 
However useful they appear, robust criteria for final confirmation is necessary to build trust among participants and maintain systemic sustainability. 
For this reason, we develop a dominant consensus mechanism to reach the final stage of agreement which is due to the dual system of both DAG and boxchain ledgers.
\begin{remark}[The dual-ledger system]
The dual-ledger system consists of a DAG in primal layer and the corresponding boxchain in the dual layer.  
\end{remark}

\subsection{The Chain of Antichains}\label{tip2}
\subsubsection{Consensus protocol for the final confirmation with box-closing, and a subtree problem}\label{consensusprotocol}
  As many practitioners and researchers have already pointed out (see, e.g., \cite{tangle}, \cite{byteball} etc.), there can be a myriad of subtrees, and as a result, a DAG can be extremely wide and inefficient in terms of the final confirmation possibilities. 
 In Figure \ref{fig61}, there is no path from some latter blue nodes (good-standing nodes) to red nodes (malignant ones). 
 With subtrees, nonsensical transactions can be ``issued and approved'' by the group of red nodes in the absence of prior investigation of the validity of transactions. This is a major obstacle to reaching a ``consensus'' on the validity of previous transactions, which is critical in any distributed peer-to-peer system.

\begin{figure}[ht!]
\begin{center}
\includegraphics[width=0.80\textwidth]{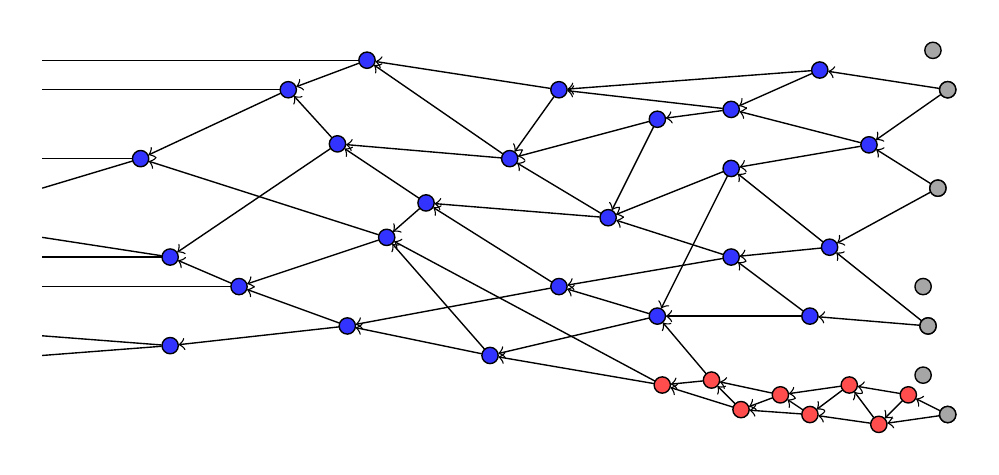}
\end{center}
\caption{Main subtree in blue vs. dishonest subtree in red. (tips are in gray.)}\label{fig61}
         \end{figure}

A subtree can occur haphazardly in a DAG, which is prone to unproductive behaviors. Regardless of the intention of forming a subtree, it is imperative to equip a channel to verify the validity of transactions from subtrees. This can be done at the algorithm-design level of the consensus mechanism of reaching decentralized final confirmation of already-approved transactions.

A subtree can take place anytime and a subtree always exists in a DAG, which renders negative behaviors possible.
Whether  it's formed by a good-will or not (by a good citizen or by a malicious attacker), there must be a way to check transactions' validity from subtrees.
 Allow us to present what can be done by our particular algorithm regarding a consensus mechanism. 
This is about reaching a final state of agreement from the peers -- the final confirmation of already-approved-transactions.

The main process of our consensus protocol is as follows. For the final confirmation of transactions of box (antichain) $B_{i-1}$,  it is required to reach agreement among the peers in the next box $B_{i}$. 
This means, a final confirmation process for $B_{i-1}$ initiates when the boxer of $B_i$ is selected. 
To join $B_i$ a node must approve at least one node from $B_{i-1}$. Thus, the formation of a box is totally based on what transactions the node approved.

\begin{remark}[The rule for the tip selection, or transaction selection and validation)]\label{ruletip}
The users are recommended to select most recent transactions. We use the rank function (in Definition \ref{rankfunction}) for this rule as the following.
\end{remark}\label{tipselection}

\begin{equation}\label{selectionrule}
r(v_k) \leq r(v_{k+1}) \leq r(v_k) + 1,
\end{equation}
where $v_k$ denotes a node that just finished its transaction validation right before $v_{k+1}$'s completion of validation. 
The inequality (\ref{selectionrule}) is equivalent to: 
\begin{equation}\label{selectionrule1}
v_k \in B_i \longrightarrow v_{k+1} \in B_i \text{ or } B_{i+1}.
\end{equation}
Once the final confirmation process begins, a later node can no longer select and approve transactions from the corresponding box whose transactions are being checked.

\begin{figure}[ht!]{\centering
         \subfigure[Primal layer: a DAG]{\includegraphics[width=0.45\textwidth]{flow0_0.pdf}}\qquad
         \subfigure[Dual layer: a boxchain (chain of antichains)]{\includegraphics[width=0.50\textwidth]{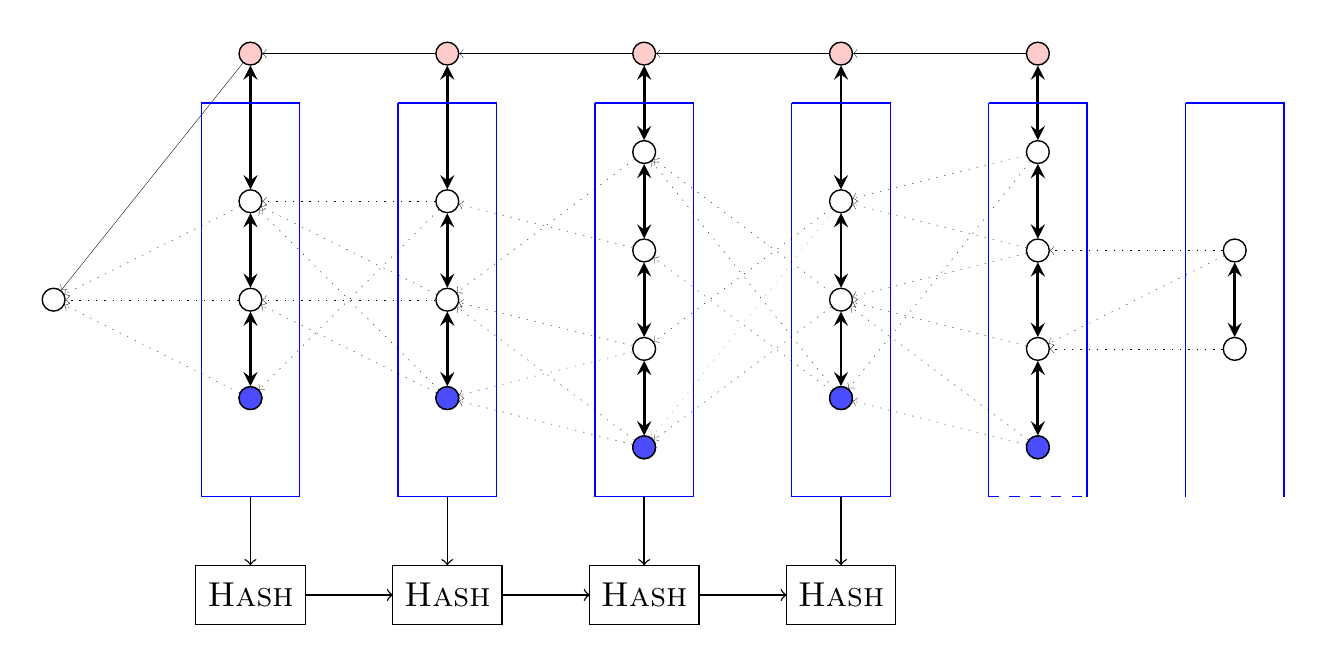}}
       	  \caption{The dual ledger-keeping. Boxers are in blue in both layers; box-geneses are in pink in the dual layer.}  \label{dualblock} }
         \end{figure}

The final confirmation will be placed on each box by its child box (the next neighboring box, or cover). Again, the final confirmation process for $B_{i-1}$ begins when the boxer of $B_i$ is selected. We, thus, list below the following {\it dual-criteria} for boxer selection. 
\begin{remark}[Criteria for a boxer selection]\label{dualc}
\end{remark}
\begin{itemize}
\item {\it Cardinality}: As mentioned in (\ref{nm}), the size of antichains can be determined by the number of nodes. 
Let $M$ denote the upper limit of the number of elements of the $i$-th box ($B_i$). Then we write $|B_i| \leq M$. If we only take this into account for a boxer selection, then the $M$th node in $B_i$ will be designated as the boxer. 

\item [-] For security purpose, $M$ will be selected as follows.
Recalling the principle of the inversion method, we compute the value $M= F_X^{-1}(p)$, where $F$ is the cumulative distribution function (c.d.f.) of any discrete random variable $X$ (or a discretized continuous random variable), $l \leq X \leq u$ where $l$ and $u$ are both positive integers; and $p$ the generated random number around the uniform distribution on $[0,1]$: It is known that if $U$ is a uniform (0,1) random variable, then with a given c.d.f.$F$, $F^{-1}(U)$ is a random variable.

\item {\it Duration}: Let the time $\tau$ denote the acceptable and agreeable time-limit (e.g., 20 seconds) for a user to wait until the final confirmation. Note that the average time until the final confirmation would be about $1.5 \tau.$ (See Remark \ref{timec} for details.) Suppose that the first node $v_{i_1}$ of box $B_i$ comes in at time $t = s$. In other words, at time $t=s$ a node  $v_{i_1}$ just finished its validation of two previous transactions in $B_{i-1}$ (or one in $B_{i-1}$ and another one in $B_{i-2}$). Suppose that a node $v_{i_k}$ finishes the validation of two previous transactions at time $t \leq s+\tau$, with the next one $v_{i_{k+1}}$ completing its approvals at time $t > s+\tau$. Then the node  $v_{i_k}$ will become the boxer and $v_{i_{k+1}}$ will be the first member of $B_{i+1}$. In fact the node $v_{i_{k+1}}$ cannot select any transactions from $B_{i-1}$ by (\ref{selectionrule1}). 
\item [-] No more nodes can be added to the box after whichever criterion is met earlier. The final confirmation process needs to get started when a boxer is selected.
\item [-] {\bf Exception:} If the size $M$ is reached at beyond allowable speed (i.e., in case that incoming transactions are extremely frequent than normal), the time constraint will be used. This is to keep the system safe from possible attacks. (see Section \ref{worst}). 

\end{itemize}

Note that the size of boxes (antichains) depends on the boxer selection problem as in Remark \ref{dualc}.
Transactions issued by the members of $B_{i-1}$ are approved by the members of $B_i$, but they are just individual validations.
We need a strong consensus protocol in order to place the final confirmation of entire set of transactions of the previous box.
\begin{table}[htp!]\nonumber
\centering{
\begin{tabular} { l }
  \hline
  {\bf Algorithm 3} The final confirmation of the transactions in $B_{i-1}$, conducted by $B_{i}$\\ 
\hline
Step 1. Once the box $B_i$'s size becomes $M$ or time is up ($\geq \tau$), i.e., $B_i$ is full,\\
\qquad \quad \  \  the $i$-genesis ($g_i$) starts to check its siblings' validations (i.e., validity of transactions of $B_{i-1}$).\\
 \\
Step 2. If all transactions of  $B_{i-1}$ are legitimate based on the consensus protocol (below), go to Step 4.\\
\qquad \quad \  \ Otherwise, go to Step 3.\\
\\
Step 3. Disable the nodes with illegal transactions in $B_{i-1}$, and report it to the genesis group.\\
\qquad \quad \  \  Disable all nodes in $B_i$ that approved such illegal nodes. \\

  \\
Step 4. The final confirmation is placed by the $g_i$. \\
\qquad \quad \    \ $B_{i-1}$ is closed and marked ``true,'' and every transaction issued in $B_{i-1}$ becomes final,\\
\qquad \quad \    \ followed by the timestamp synchronization with all the previous box-geneses.\\

 \hline
\end{tabular}
\label{table2info}
}
\end{table}

Therefore, the whole boxchain (chain of antichains) is marked ``true" in a recursive way.  
Note that there is another recursion inside each box as in the following consensus mechanism.\\

\noindent
\underline{\bf Consensus protocol in a dual recursion} 
\begin{itemize}
\item The 2+2 process in recursion:\\
When a node approves two previous transactions it will be added to a box as in the Algorithm 2. As the node joins the box, it is automatically augmented to the last node in the box and is required to check the validity of transactions approved by the neighboring node.
This process continues recursively from the second to the last node (boxer) of the box. Simply put, each node checks its prior neighbor's approvals upon arrival. This is a dual validation process for every node conducted both in primal and dual layers.

\item Final confirmation by the box-genesis:\\
The final confirmation of recent transactions will be made by the box-genesis as presented in Algorithm 3. This role resembles that of miners' of Bitcoin, and such validation process is conducted by the box-genesis that only exists in the dual layer. Rewards are given to a box-genesis when the final confirmation is placed.
\end{itemize}

The final confirmation is placed when the consensus protocol is completed.
The box-genesis broadcasts the termination of validation process to all other members. 
In our peer-to-peer system anyone can take such role (not the miners vs. the rest) as long as one shows a good track-record. 

\begin{remark}[The box-genesis group membership eligible to good-standing nodes]\label{genesisrandom2}
The roles of box-genesis is crucial to keeping our ecosystem healthy and secure. 
Excellent history of positive behavior is required for peers to renew this membership and new members with outstanding track-records can also be selected.
The box-genesis among the leaders of consensus protocol, and this position is open to everyone but follows a random assignment process.
\end{remark}

\begin{remark}[Random assignment process]\label{genesisrandom}
The members of a box-genesis group are not totally trusted even though they qualify -- its box assignment is completely random with some conditions, and not known to the box-genesis until the last node of box (boxer) is determined. After the boxer is selected the box genesis will be appointed by a random selection among the nodes other than the boxer.
\end{remark}

\begin{remark}[Boxer's role for dual safety]
The boxer is the one communicating with its siblings (nodes of the box), together with other boxers.
However, the final confirmation of validations of previous transactions must be determined solely by a box-genesis. 
For dual safety, the boxer keeps the final confirmation from the box-genesis in check. Some rewards are given to the boxer. 
\end{remark}

\begin{remark}
As we discussed, our recursive consensus mechanism can perform well in terms of both efficiency and scalability. 
This consensus protocol is unique -- in our dual-ledger ecosystem there are no heterogeneous peer groups. 
Anyone qualified and ``lucky'' node can play a vital role in the system, e.g., as a box-genesis or as a boxer.
 It is easy to see that more active peers will have higher chances to get more rewards. 
\end{remark}

\subsubsection{Possible attack to the system and its success probability}\label{worst}

A double-spending attack can cause a serious integrity violation in any distributed peer-to-peer system. Although our recursive validation and final confirmation systems are robust, every possible attack needs to be accounted for in order to keep the network safe. A possible and worst-case attack scenario can be imagined as follows. Suppose that a malicious and affluent party attempts to take over the system over some time-window by issuing an immense number of transactions almost simultaneously with all different addresses. In our dual ledger-keeping system, a complete dominance over a pair of back-to-back boxes (antichains) including the box-geneses and the boxers is necessary for a meaningful damage to transactions.

From the attacker's perspective, the favorable event is that there are no other transactions issued while two consecutive boxes are formed. 
Suppose that a time constraint $\tau$ (in Remark \ref{dualc}) is used for a single box-closing. 
For simplicity we also assume that the transaction arrivals follow a homogeneous Poisson process. More detailed and realistic, nonhomogeneous compound Poisson processes, etc. are covered in Section \ref{stoc}.
It is known that the corresponding inter-arrival times are i.i.d. exponential random variables with the same parameter $\lambda$ of the Poisson process (the average number of transactions in a given unit time). 
Let a random variable $W$ denote the time until the next transaction, which follows an exponential distribution with a p.d.f. $f(x) = \lambda e^{-\lambda x}, x \geq 0$.
Then the following can be written:
\begin{equation}\label{p1}
p= P(\text{no transactions while two boxes are formed}) =  P(W > 2\tau) = e^{- 2\lambda \tau}.
\end{equation}
Suppose that the system had set a time limit $\tau$ of 20 seconds (1/3 minute) for the final confirmation process. Assuming the average number of transactions per minute is 30 (i.e., $\lambda = 30$), we have  
\begin{equation}\label{p2}
p= P(W > 2(1/3)) = e^{- 2(30)(\frac{1}{3})} \approx 0.000000002061.
\end{equation}
For a more speedy final confirmation process, let  $\tau = 10 \text{ seconds}$ with the same $\lambda = 30.$
\begin{equation}\label{p3}
p= P(W > 2(1/6)) = e^{- 2(30)(\frac{1}{6})} \approx 0.00004539,
\end{equation}
which is still a sufficiently slim chance. 
The results clearly demonstrate the following.
\begin{itemize}
\item More transactions will lower the probability of a successful attack, i.e., $\lambda \uparrow \Longrightarrow p \downarrow.$
\item A longer box-closing interval also brings down the chance of a successful attack, i.e., $\tau \uparrow \Longrightarrow p \downarrow.$
\end{itemize}

Note that the number of transactions will be sizable in a booming e-commerce marketplace, which makes it enables a very speedy final confirmation process. 
With this reasoning, let us consider another case, e.g., say 100 transactions per minute on average (i.e., $\lambda = 100$).
Then our time constraint $\tau$ can be selected so that the probability of the attacker's favorable event is kept negligible. 
Let us put an upper bound of $p = 0.000001$ (one out of a million) on a successful attack chance. 
Then we can find $\tau$ using
\begin{equation}\label{p4}
 P(W > 2\tau) = e^{- 2(100)\tau} \leq 0.000001,
\end{equation}
followed by
\begin{equation}\label{p5}
\tau \geq \frac{ \ln (0.000001)}{-200} \approx 0.06907755 \text{ minutes},
\end{equation}
which is about 4.14465316 seconds. This means that if we put on a 5 second rule for the time constraint $\tau$ then the attacker's success probability would be less than 0.000001 (again, $\tau \uparrow \Longrightarrow p \downarrow$). If it still seems possible, note that we have a series of box-genesis (two different ones) in the associated back-to-back boxes and they are independent of transactions.

\begin{remark}\label{genesisrandom1}
Together with a random selection of the box-genesis, the probability of successful attack would be almost zero in any case. 
\end{remark}

Detailed stochastic aspects of transaction flows will be discussed in Section \ref{stoc}, where one can find more general and useful ideas about how to analyze related random processes.

\begin{remark}[Average time until the final confirmation]\label{timec}
It is about $1.5\tau$ which is the average of $\tau < T < 2\tau$.
\end{remark}
The actual amount of time until the final confirmation for a singe transaction (or a node) is varied depending on the time point when the node joined its  corresponding box. 
Such waiting time until the final confirmation is $\tau < T < 2\tau$ since the final confirmation is made by the next neighboring box (child box). Note that it takes $\tau$ to form a single box.  It is not hard to see that if a node is the first member of a box, its transaction's final confirmation time would be approximately $2\tau$ since two more boxes need to be closed.
 In case of a boxer, the confirmation time will be about $\tau.$ 
Recall that the transactions are already approved by their child nodes in a DAG, and therefore our consensus protocol in the dual layer can place a powerful final-confirmation.

\begin{remark}
In the dual layer, we have a chain of antichains (boxchain) and each antichain (box) consists of a single chain. This means we have a total order in the dual layer and obviously any conflicting nodes can easily be checked. 
\end{remark}

\noindent
Let us finalize this section by supporting Remarks  \ref{ruletip} and \ref{dualc} which enable the following.
\begin{itemize}
\item No purely random approvals of tips - a lazy user can approve a fixed pair of old transactions

\item Inclusion in a more recent antichain, followed by a 2+2 validation process.

\item In case that a new legitimate transaction is not approved and waits longer than the formation of two new boxes, an empty transaction will be issued as a series. 
 But it still needs to select and approve another node to issue the same transaction in order to get a validation by others. Note that such empty transaction contributes to the network's security. 

\item $M$ - the width constraint for antichain: This will be determined by the frequency of transaction flows (see Section \ref{stoc} for more details about a random transaction flow).

\item $\tau$ - the time constraint for antichain: We do not wait until the box is filled up to the given limit $M$. Recall: If the size $M$ is reached at beyond allowable speed (incoming transactions are extremely frequent than normal), the time constraint had better be used. This is to keep the system safe from the worst case scenario we just studied.
\end{itemize}

\subsubsection{Incentive system with proper rewards and fees}\label{incentivesystem}
Proper rewards and fees are essential to keep a distributed peer-to-peer network at a desirable status. Note that in Bitcoin ``mining'' is the incentive mechanism for decentralized security. Our incentive device is for every participant as a mixture of rewards for contribution and fees for usage. The fees in our ecosystem are minimal but serve as an important obstacle to nonsensical orders. Below are the possible cases for rewards. Rewards are given to such user(s) who

\begin{itemize}
\item put legitimate validations in the DAG-based primal layer.
\item put legitimate validations in the boxchain-synchronizing dual layer.
 \item complete the job as a box-genesis 
 \item complete the job as a boxer
\item report abnormal events
\end{itemize}

Recall that the appointment as a boxer or a box-geneis is conditionally at random -- any qualified user can play these roles multiple times (as mentioned preciously, e.g., see Remarks \ref{boxersrole}, \ref{genesisrole}, \ref{genesisrandom2}, \ref{genesisrandom}). It is not difficult to see that more active and honest participation will increase the chance to take such positions and receive more rewards, which could end up with more rewards than fees.  
Fees are paid whenever a transaction is issued. This incentive system is certainly beneficial to both users and the whole ecosystem. Good citizenship is crucial for our purely distributed peer-to-peer system as much as any type of society is in need of.

\subsubsection{Our dual ledger-keeping system, boxchain compared with others}\label{tip3}

The main structure of a blockchain is a single serial link of blocks (hence the name) and each block consists of transactions, therefore the blockchain is a totally ordered set. If we look at each transaction as a single block, there is no longer a single chain; Rather, we see very disorganized, myriads of intertwined blockchains. This generalized blockchain and can be depicted as a DAG, consisting of a set of vertices (transactions) and a set of edges (issue and validation of transactions). The transactions in a DAG are only partially ordered. Indeed, the validation process in a DAG is very efficient and relatively but not as much secure as in the original blockchain. Considered a third generation blockchain after Ethereum as the second, which emerged after blockchain, a number of developers have recently released DAG-based distributed ledgers, associated cryptocurrencies, and applicable payment systems, led by IOTA, Hashgraph, etc.

 Irregular and partially ordered DAGs  can be reshaped into a more compact form using the discrete mathematical notions of chains, antichains, and SDRs. This idea, coupled with a doubly-secure  confirmation protocol, has inspired our dual ledger-keeping system, which is partially based on  \cite{union2}, \cite{lkp},  \cite{union1}. While preserving the DAG-based ledger, we build another ledger with a chain of antichains to structurally resemble the original blockchain. As a result, along with an effective consensus protocol, our boxchain is an efficient, secure, and decentralized DLT.

\begin{table}[htp!]
\centering{
\begin{small}
\begin{tabular} { l | c c c }
  \hline
 &Bitcoin &  IOTA & Box Protocol \\ 
\hline\hline
Usage  & Payment	& Used for IoT applications   &Payment\\
\hline\hline
Finality	&Yes	&  No & Yes\\
\hline\hline
Confirmation	&\multirow{2}{*}{10 minutes at least }	 &\multirow{2}{*}{2 minutes (validation)} & 30 seconds on average if $\tau = 20$ seconds \\
time                   & 										&										&(depending on the values of $\tau$ and $\lambda$)\\
\hline\hline
Transaction fees & 0.001 BTC	&  None & Tiny or none (if rewards are given)\\
\hline\hline
TPS&	7 &  1000 & 5000\\
\hline\hline
Decentralization	&	Yes & No & Yes\\
\hline\hline
Security	&high 	&  low & high  \\
 \hline
\end{tabular}\caption{Brief Comparison of Bitcoin, IOTA and Box Protocol}\label{table1info}
\end{small}
}
\end{table}

\section{Stochastic models of the dual ledger-keeping system}\label{stoc}

\subsection{A nonhomogeneous Poisson process - practical assumptions on the frequency of transactions}

Unlike some of the assumptions made in \cite{tangle}, we suppose that $X(t)$ is a nonhomogeneous P.P. (Poisson Process), where $\{\lambda(t), t\geq 0\}$ is a stochastic process itself. This is because we believe that transactions' frequency can be varied over certain time periods (e.g., there might be more transactions in the middle of night).
For the sake of completeness, allow us to present some basic concepts first. Stochastic process can simply be called ``one-parametric family'' of random variables $X(t), t \in T$. For example, in case of Markov Chains we have $T = \{ 0, 1, 2, \dots  \}.$ In our case we suppose $ T = \{t \ | \ t\geq 0\},$ which means ``time'' (mathematically it is nothing but nonnegative half of the real line).
The following is well known. $X(t), t \geq 0$ is a homogeneous Poisson process if it has two properties: (i) $X(t)$ has independent increments, i.e., for $0 \leq t_0 < t_1 \leq t_2 < t_3 \leq \dots \leq t_{2n -1} < t_{2n}$,
$X(t_1) - X(t_0), X(t_3) - X(t_2), \dots, X(t_{2n}) - X(t_{2n-1})$ are independent random variables; (ii) For $s \geq 0, t > 0$, the random variable has Poisson distribution with parameter $\lambda t$,
$$P(X(s+t) - X(s) = k) = \frac{(\lambda t)^k}{k!} e^{- \lambda t}, k = 0, 1, \dots,$$
which follows that $E(X(s+t) - X(s)) = \lambda t$ and $var(X(s+t) - X(s)) = \lambda t.$ Poisson process is a random event process, e.g., transactions occur in the system (our case), customers arrive to a store, cars arrive to a gas station, etc. $\lambda$ is the expected number of such random events in unit time. 

For nonhomogeneous Poisson process, Condition (ii) only needs to be changed. Instead of a constant rate $lambda$ we assume that there exists a nonnegative function $\lambda(u), u \geq 0$ such that for $s < t$ $X(t) - X(s)$ has Poisson distribution with parameter
$$\int_s^t \lambda(u) du,$$ 
which means the parameter depends on time interval. This assumption is crucial in our business because our data set shows some clear trends in frequency of transactions (and their sizes as well). $\lambda(u)$ is called ``event density'' (in our case, transaction density) because in the interval $(t, t + \Delta t)$
$$ \int_{t}^{t + \Delta t} \lambda(u) du \approx \lambda(t) \Delta t$$
is the expected number of events.

\begin{example}[Transaction Density from Data Set]
Suppose that the following transaction density (or event density) are found from recent data set (unit time = 1 hour).  
$$\lambda(u) = \left\{ \begin{array}{ll}  2u, & 0 \leq u \leq 1\\
							 2, & 1 \leq u \leq 2\\
							 4-u, & 2 \leq u \leq 4.  \end{array}\right.$$
For simplicity, we are interested only in the time interval [0, 4]. We want to calculate 
\begin{itemize}
\item [(i)] the probability that two transactions occurred during the first two hours, 
\item [(ii)] the probability that two transactions occurred during the second two hours.
\end{itemize}
\end{example}
For (i), we need to calculate the event density, say $\mu$, of the first two hours:
$$\mu = \int_0^2 \lambda(u) du = \int_0^1 \lambda(u) du + \int_1^2 \lambda(u) du = \int_0^1 2u du + \int_1^2 2 du = 1+2 =3.$$
Thus, $P(X(2) = 2) = \frac{3^2}{2!} e^{-3} = 0.2240.$
For (ii), we have:
$$P(X(4) - X(2) = 2) = \frac{\left(\int_2^4 \lambda(u)du\right)^2}{2!}e^{- \int_2^4 \lambda(u)du} = \frac{\left(\int_2^4 (4- u) du\right)^2}{2!}e^{- \int_2^4 (4-u) du} = \frac{2^2}{2!}e^{-2} = 0.2707. $$

\begin{remark}[Doubly stochastic Poisson processes] A nonhomogeneous Poisson process with the rate function $\{ \lambda(t), t \geq 0 \}$ -- it is a stochastic process itself -- is called a ``doubly stochastic Poisson process.'' This is a great fit for applications which have ``dependent" process increments, e.g., some seasonal products, and new IT products (with business cycles, replaced by newer products), etc. 
\end{remark}
The simplest doubly stochastic process (sometimes called a mixed Poisson process) has a single random variable $\theta$ with $X'(t) =  X(\theta t),$ where $\{ X(t), t \geq 0\}$ is a Poisson process with $\lambda = 1$.
Given $\theta$, $X'$ s a Poisson process of constant rate $\lambda = \theta$, but $\theta$ is random (unobservable, typically). If $\theta$ is a continuous random variable with p.d.f. $f(\theta)$, the marginal distribution is as follows.
$$P(X'(t) = k) = \int_0^\infty \frac{(\theta t)^k e^{-\theta t}}{k!}f(\theta)d\theta.$$

\subsection{Some important random behaviors regarding transaction interarrival times in a Poisson process}\label{inter}
It is well known that inter-arrival times, let's say $S_0, S_1, \dots$, are i.i.d. exponential random variables with parameter $\lambda$, where $\lambda$ is the parameter of the Poisson process.
Let $W_n$ denote the time of occurrence of the $n$th event (setting $W_0 = 0$). The differences $S_n = W_{n+1} - W_n$ are the duration that the Poisson process sojourns in state $n$. 
See Figure \ref{fig1} for description, where $S_0, S_1, \dots$ denote the inter-arrival times in a Poisson process and $W_1 = S_0, W_2 = S_0 + S_1, \dots$ represent the occurrence times of the random events in a Poisson process.

\begin{figure}[ht!]
\begin{center}
\includegraphics[width=0.80\textwidth]{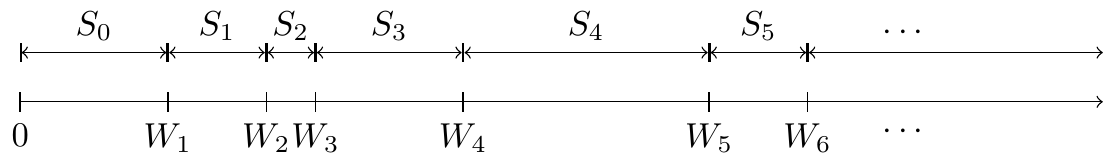}
\end{center}
\caption{Interarrival times of transactions - exponential random variables in the Poisson process.}\label{fig1}
         \end{figure}

We present the following theorems without proofs (see some stochastic models literature for more details).
\begin{theorem}
Conditioned on $N(t) = n$, i.e., $N(t) = N((0, t)) =$  number of events in $(0, t)$, the random variables $W_1, \dots, W_n$  have joint p.d.f.:
$$g(t_1, \dots, t_n) = \left\{\begin{array}{ll} \displaystyle \frac{n!}{t^n}&\text{if } 0 \leq t_1 \leq \dots \leq t_n\\ [2ex] 0 & \text{ elsewhere} \end{array}\right.$$

\end{theorem}

The above theorem tells us that if conditioned on the number of events up to $t$, i.e., $N(t) = n,$ then the $n$ points behave as $n$ independent random points, each chosen from $(0, t)$, according to uniform distribution.

\begin{theorem}\label{thm2}
Probability density function of sum of $n$ independent, exponentially distributed random variables with the same parameter $\lambda$: $X = X_1 + \dots + X_n,$ where $X_i, i=1, \dots, n$ have the p.d.f. $f(x) = \lambda e^{-\lambda x}, x\geq 0.$ Let $f_n(x)$ be the p.d.f. of $X$. Then, for $n=1, 2, \dots$, we have
$$f_n(x) = \frac{\lambda^n x^{n-1} e^{-\lambda x}}{(n-1)!}, x \geq 0.$$ 
\end{theorem}

It follows that another derivation of probability distribution for interarrival times of transactions.
Let $F_n(t) = P(S_0 + \dots +S_{n-1} \leq t)$. Then we have, using Theorem \ref{thm2},
$$F_n(t) = \int_0^t f_n(u) du = \int_0^t  \frac{\lambda^n u^{n-1} e^{-\lambda u}}{(n-1)!} du.$$
Furthermore, for $n \geq 1$,
\begin{equation}
\begin{array}{lcl}
P_n(t) = P(N((0, t))=n) &=& F_n(t) - F_{n+1}(t)\\[2ex]
&=& \displaystyle \int_0^t  \frac{\lambda^n u^{n-1} e^{-\lambda u}}{(n-1)!} du - \int_0^t  \frac{\lambda^{n+1} u^{n} e^{-\lambda u}}{n!} du\\[2ex]
&=& \displaystyle \left[\frac{\lambda^n u^{n} e^{-\lambda u}}{n!}\right]_0^t+ \int_0^t  \frac{\lambda^{n+1} u^{n} e^{-\lambda u}}{n!} du - \int_0^t  \frac{\lambda^{n+1} u^{n} e^{-\lambda u}}{n!} du\\[2ex]
&=& \displaystyle \frac{(\lambda t)^n}{n!} e^{-\lambda t}.
\end{array}
\end{equation}
Note that $F_k(t)$ means the probability that there are at least $k$ transactions. 
It follows that $P_n(t) =  \displaystyle \frac{(\lambda t)^n}{n!} e^{-\lambda t}$ for $n=1,2, \dots$ and this implies $P_0(t) = e^{-\lambda t}.$
See our related application below.

\begin{example}[Expected total boxcoin in time $(0, t)$]
Suppose that transactions occur in the system according to a Poisson process. Upon making a transaction each user pay 1 boxcoin. Find the expected total sum collected in $(0, t),$ discounted back to time 0. Let $\beta$ be the discount rate. Note that the reason of discounting is to perform more accurate comparison, especially for different time periods.
\end{example}

In case of continuous compounding $\beta$ is divided by $n$ and then $n \rightarrow \infty$. The result is 
$$\lim_{n \rightarrow \infty} \left(1 + \frac{\beta}{n}\right)^n = e^\beta.$$ The discounting factor is its reciprocal value $e^{-\beta}$. Discounting from time $W_k$ back to initial time 0, the result is $e^{-\beta W_k}.$
We want to calculate the expected total boxcoin in a given time as follows.
\begin{equation}
\begin{array}{lcl}
\displaystyle E\left(  \sum_{k=1}^{X(t)} e^{-\beta W_k}  \right) &=& \displaystyle \sum_{n=1}^\infty E\left(  \sum_{k=1}^{X(t)} e^{-\beta W_k} \big| X(t) = n \right)P(X(t) = n)\\
											&=& \displaystyle \sum_{n=1}^\infty n E ( e^{-\beta W_1})P(X(t) = n)\\
											&=& \displaystyle \sum_{n=1}^\infty n \int_0^t \frac{1}{t} e^{-\beta u}du P(X(t) = n)\\
											&=& \displaystyle  \frac{1}{\beta t} (1- e^{-\beta t})      \sum_{n=1}^\infty n P(X(t) = n)\\
											&=& \displaystyle  \frac{\lambda}{\beta} (1- e^{-\beta t}).		
\end{array}
\end{equation}

It is well known that if $X_1, \dots, X_n$ are i.i.d. exponential random variables, then $X_1+\dots +X_n \sim \text{Gamma} (n, \mu),$ which has the p.d.f.
\begin{equation}\label{gamma}
f(x)  = \mu e^{-\mu x} \frac{(\mu x)^{n-1}}{(n-1)!},
\end{equation}
where $(n-1)! = \Gamma (n) $ because $\Gamma(n) = \int_0^\infty e^{-y} y ^{n-1} dy$ and $\Gamma(n) = (n-1)\Gamma(n-1)$.  Its expectation and variance are $E(X) =n/\mu, \text{Var}(X) = n/\mu^2,$ respectively. Note that the gamma distribution is log-concave, so we can formulate a convex optimization problem using some stochastic optimization techniques. \\

\noindent
{\bf Gamma distribution, a family member of logconcave distributions.}
If $\mu = 1$ of (\ref{gamma}), then the distribution is said to be standard. If $\xi$ has gamma distribution, then $n \xi$ has standard gamma distribution. Both the expectation and the variance of a standard gamma distribution are equal to $n$.  
An $m$-variate gamma distribution can be defined in the following way.  Let $A$ be the $m \times (2^m -1)$ matrix the columns of which are all 0-1 component vectors of size $m$ except for the 0 vector. Let $\eta_1, \dots, \eta_s, \ s=2^m -1$ be independent standard gamma distributed random variables and designate by $\eta$ the vector of these components. Then we say that the random vector
\begin{displaymath}
\xi = A \eta
\end{displaymath}
has an $m$-variate standard gamma distribution.

\subsection{Combining nonhomogeneous Poisson processes of transactions}\label{inflow}
Multiple transaction inflows into the system must be well analyzed in order to help monitor overall activity of the system.
Let us begin with the simplest case and extend to the generality. Let $X_1(t)$ and $X_2(t)$ be the two Poisson flows with parameter $\lambda_1$ and $\lambda_2$, respectively. We also assume that they are independent. Consider nonoverlapping intervals: $(t_0, t_1)$ and $(t_2, t_3).$ Then $X_1(t_1) - X_1(t_0)$ and $X_1(t_3) - X_1(t_2)$ are independent, and $X_2(t_1) - X_2(t_0)$ and $X_2(t_3) - X_2(t_2)$ are also independent. 
Since the two processes are also independent, the four random variables are independent. It follows that 
$X_1(t_1) - X_1(t_0) + X_1(t_3) - X_1(t_2)$ and $X_2(t_1) - X_2(t_0) + X_2(t_3) - X_2(t_2)$ are independent. Thus, we can say that $X(t) = X_1(t) + X_2(t)$ has independent increments. Note that
$$X(s+t) - X(s) = X_1(s+t) - X_1(s) + X_2(s+t) - X_2(s),$$
where $X_i(s+t) - X_i(s)$ has Poisson parameter $\lambda_it, i= 1, 2.$ Hence $X(s+t) - X(s)$ has Poisson distribution with parameter $(\lambda_1 + \lambda_2)t.$

The procedure is the same in the general case. We can unite arbitrary finite number of independent Poisson process: if $X_1(t), \dots, X_n(t) $ are independent Poisson processes with parameters $\lambda_1, \dots, \lambda_n$, respectively, then $X(t) = X_1(t)+ \dots + X_n(t)$ is a Poisson process with parameter $\lambda_1 + \dots + \lambda_n.$

\subsection{Transaction ``size", and the recursive formulas for the probability mass function of compound random variables}

Transactions are randomly occurring events. Since a number of events occurring in a fixed period of time, say $N$, is uncertain, it would be suitable to use a compound distribution modelling for a random sum $S =X_1+X_2+ \cdots +X_N$, where $N$ is a nonnegative integer-valued random variable.  
For the frequency of transactions, let $N(t)$ be a Poisson process with a fixed rate $\lambda$, for simplicity. 
It is not a problem, however, if the Poisson parameter $\lambda$ is also uncertain as discussed in the previous sections. 

If we need to consider the magnitude (or size) of transactions but their sizes are random (or unknown), it is suitable to use a compound Poisson process.
The following is well-known (\cite{gammapoisson}):
$$N \sim \text{Poisson} (\lambda) \text{, where }  \lambda \sim \text{Gamma} (r, p/(1-p)) \rightarrow N \sim \text{NegBinomial} (r, p),$$
which means that if the Poisson parameter $\lambda$ is uncertain due to the behavior of heterogeneous users, a gamma distribution may be suitable for capturing the $\lambda$ information. Then $N$ will have a negative binomial distribution. Negative binomial distributions can be used for our model with no problem, but we restrict ourselves to the compound Poisson distribution for this paper.

The frequency of transactions plays a central role and in our view, the size of transactions would also be very important to proper management of a DAG. For the transaction size, we let $X_1, X_2, \dots$ denote independent, identically distributed random variables (meaning the transaction sizes at the corresponding events), which are also independent of the Poisson process. (i.e., the random variables $N, X_1, X_2, \dots$ are mutually independent.)  
Then $S(t) = \sum_{k=1}^{N(t)} X_k, \forall t \geq 0$  (i.e., $S(t)$ is a compound random variable) means the aggregate transaction weight in time interval $(0, t)$.
It is well known that $E(S(t)) = \lambda t \mu, var(S(t)) = \lambda t (\sigma^2 + \mu^2)$ if $\mu = E(X_k), \sigma^2 = var(X_k), k = 1, 2, \dots$.

In order to calculate the distribution of aggregate transaction weights we use Panjer's recursion formula (see for details, e.g., \cite{recursion1}, \cite{recursion2}, \cite{recursion3}, \cite{insurance}, etc.).
For completeness, allow us to present here some basic notions and related formulas. 
The size $X_k$ can be fitted with continuous or discrete distributions, but the continuous distribution needs to be discretized for the use of Panjer's recursion formula. This is a great fit for our model because every transaction has a positive integer (its weight).
For practical discretization methods we refer the reader to the literature, e.g., \cite{recursion4} and \cite{recursion3}.
We restrict our attention to the case that the transaction size $X_k$ follows a discrete distribution on the positive integers  because a continuous variable needs to be discretized to use the recursion formula and also because any positive integer-valued variables can easily be scaled to the suitable size. Then $S$ is also distributed on the nonnegative integers, and the probability mass function of the compound process $S(t)$ can be calculated recursively by the following well-known recursion:
\begin{equation}
\displaystyle P(S(t) = k) = \frac{1}{k} \lambda \sum_{i=1}^k i P(X_1  =i) P(S(t) = k-i).
\end{equation}

\begin{figure}[ht!]
\begin{center}
\includegraphics[width=0.70\textwidth]{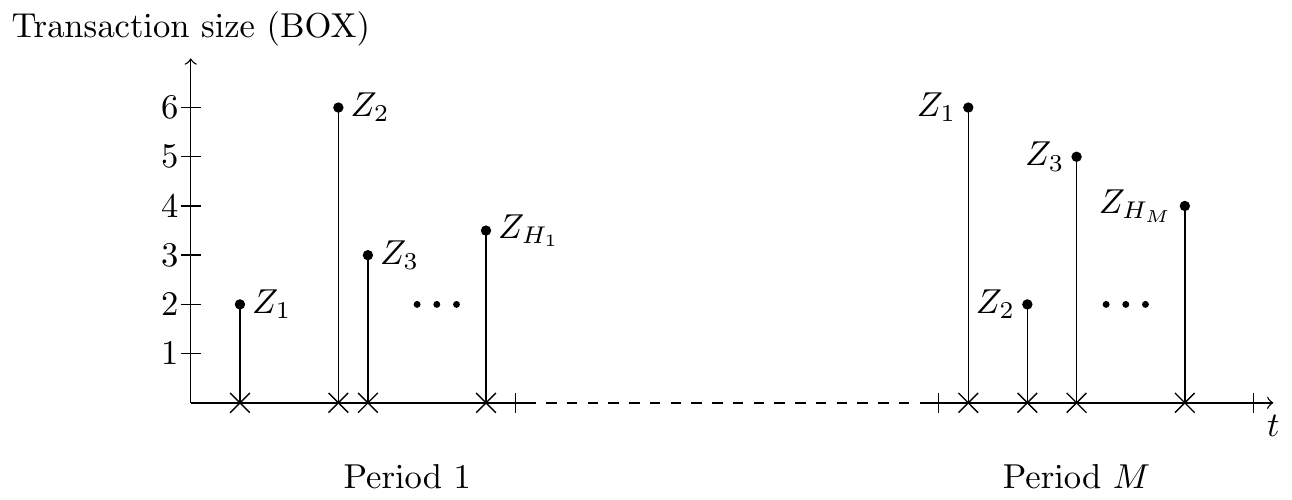}
\end{center}
\caption{Illustration of Compound Poisson Distributed Transactions. $H_j$ is the number of events incurred over the period $j$ and the unit transaction size is one boxcoin.}\label{fig2}
         \end{figure}

Given a Poisson process $N_j(t)$ and nonnegative integer-valued random transaction size $X_{jk}$'s for the upcoming time periods $j$, $j =1, \dots, M$, we can write:
\begin{equation}
\begin{array}{l}
\displaystyle P(N_j (t) = x) = \frac{(\lambda_j t)^x}{x !} e^{-\lambda_j t}, \ x= 0, 1, \dots ; \ j = 1, \dots, M\\[1.5ex]
S_j(t) = X_{j1} + X_{j2} +\dots + X_{jN_j(t)}, \ j = 1, \dots, M,
\end{array}\end{equation}
where $X_{jk}$ is the $k$th transaction size in the $j$th process. See Figure \ref{fig2} for description.

Let $f_j(x) = P(S_j(t) = x)$, where $x$ is a positive integer. Then, the recursion formula for the p.m.f $f_j(x),  j=1, \dots, M$ can be written as:
\begin{equation}
\begin{array}{l}
\displaystyle f_j(x) = \frac{\lambda_j}{x} \sum_{k=1}^x kp_j(k)f(x-k), \ x=1,2,\dots, \\[3ex]
\displaystyle f_j(0) = e^{-\lambda_j}, \ j=1,\dots,M,
\end{array}
\end{equation}
where $p_j(k) = P(X_{j1}  = k).$

\section{Models and discussions for boxdollar}\label{bxd}

\subsection{Background of boxdollar}

Boxdollar (BXD), a value-preserving medium of payment, store of value, or stablecoin, will be used primarily as a useful medium of exchange as well as a dependable store of value. The value of BXD is at a fixed exchange rate one to one to a fiat currency, hence boxdollar is a fiat money-backed asset and keeps its stability. The goal of analytical models of boxdollar is twofold { currency conversion and currency rebalancing problems. These models are well known, and therefore we list up a few models in this section. A variety of useful optimization models for financial applications, which can be found in the operations research literature (see, e.g., \cite{stochastic}, \cite{stochasticpflug}, \cite{stochastic01}, etc.).

P2P payments are among impactful applications due to the similarity of operations and transactions on a decentralized system structure: They have chains of digital signatures of asset transfers, P2P networking, etc. Yet, most of P2P platforms are not fully utilizing the ``purely" distributed system. There are numerous of P2P platforms which made it easier and faster to make payments, send or receive funds. They provide us with clearly better solutions than traditional systems in most aspects, even with lower fees and costs. Still, their pricing and operating models could be more efficient and secure by the blockchain technology, which may lead us to find a more perfect solution.

Most of existing cryptocurrencies are not accepted in regular market places due to their high volatility. Holding some cryptocurrency is extremely risky, so the current forms of cryptocurrencies are not suitable to use in our daily life. It follows, therefore, the price stabilization is the key to the application of a cryptocurrency. A cryptocurrency, with guaranteed price stability, will clearly be beneficial. The electronic commerce (or e-commerce) industry is where one can find it suitable to apply the blockchain technology and cryptocurrency to its existing P2P payment system. Most of successful e-commerce marketplaces have flexible and adaptable structure, and more importantly, with less regulation but higher potential. For this reason, the e-commerce is a great t for incorporating the ideas of cryptocurrency and its related applications. At this moment the e-commerce market transactions are mostly based on at currencies, while there has been a serious need of instantaneous transactions and borderless transfer-of-ownership. Yet, there is a variety of difficulties in the international transfer, which often challenges tracking, and normally requires more time (than domestic) with higher cost to exchange, various fees depending on a sending or receiving financial institution.

For instantaneous transactions, possible real-life applications include supply chain in drop-shipping, manufacturing management, insurance claim payment process, tax collection processes, and many others. The existing models are not well-suited for our daily use on such applications, and this inspired us of a new stablecoin.

\subsection{Currency conversion problem}

We may encounter situations of needing to exchange boxdollar (BXD) to some other fiat currency (possibly to multiple currencies) -- observation of nontrivial movement in exchange rates, our market participants' requests, and/or regular portfolio (multicurrency accounts) rebalancing, etc. When it comes to currency exchange, the loss of stabilization is what we're most concerned about. This means we should always maximize the USD value (i.e., the amount of BXD) of desired positions in converting to other currencies. It is very important that we obtain the best currency conversion strategy to keep up with our stable Box Dollar models. Such currency conversion problems (with possibility of  arbitrage detection) are well known in operations research literature.

\begin{figure}[ht!]
\begin{center}
\includegraphics[width=0.70\textwidth]{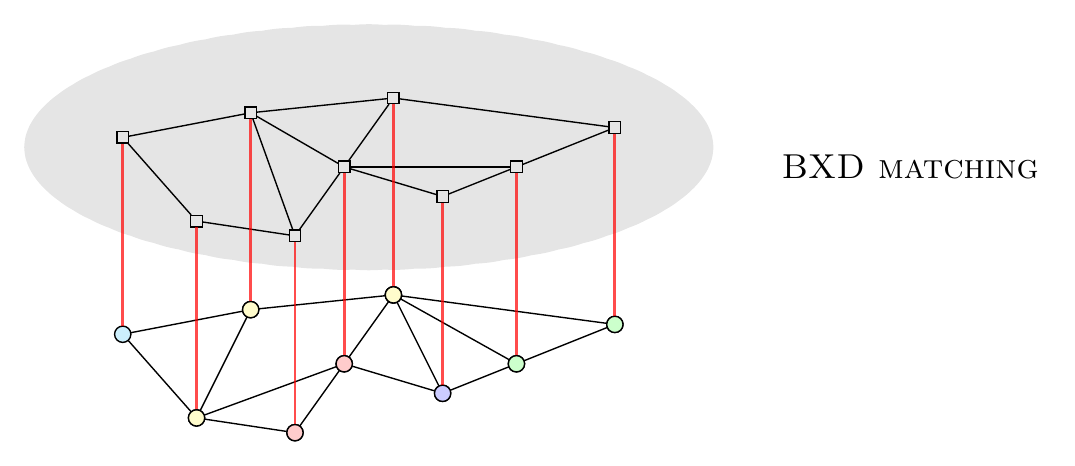}
\end{center}
\caption{Stable Coin System Topology - BXD matching with fiat currencies: CNY (red), JPY (yellow), KRW (blue), HK\$ (green), NT\$ (purple)}\label{fig0}
         \end{figure}

In what follows, we present optimization problems suitable to rebalance the combination of multinational currencies. Their optimal solutions (i.e., optimal weights on currencies) will be employed for a conversion problem.
Note that currency rebalancing will take place only if there is a guaranteed capital gain. 

\subsection{ Currency rebalancing -- Kataoka's Problem: the safety first model}\label{stp2}

Typically random variables appear only on the right hand sides of constraints of stochastic program. The following problem was formulated by \cite{kataoka} and it has stochastic constraints where their technology matrix has random variables. 
Many real life applications can be solved, but we restrict ourselves to our problem in this paper. 
Let $x_i$ denote the weight of the $i$th currency,$ i=1, \dots, n$ and the random vector $\xi$ consists of components meaning the return on holding the corresponding currency. Our model is the following:
\begin{equation}\label{k2}
\begin{array}{l}
\max d\\
\text{subject to}\\
\displaystyle P\left( \sum_{i=1}^n \xi_i x_i \geq d \right) \geq p\\
\sum_{i=1}^n x_i = 1, x \geq 0,
\end{array}
\end{equation}
where we assume that $\xi = (\xi_1, \dots, \xi_n)^T$ has an $n$-variate normal distribution with
$$\mu_i = E(\xi_i), \ i=1, \dots, n, \mu = (\mu_1, \dots, \mu_n)^T,$$
$$C = E{(\xi - \mu)(\xi- \mu)^T}.$$

Note that $p$ and $M$ are constants and the decision variables are $x_1, \dots, x_n, d$. 
Using some mathematical steps (see, e.g., \cite{bounding5}) the formulation (\ref{k2}) can be written up as
\begin{equation}\label{kataoka}
\begin{array}{l}
\displaystyle \max \left\{ \mu^Tx + \Phi^{-1} (1-p) \sqrt{x^T Cx}  \right\}\\
\text{subject to}\\
\sum_{i=1}^n x_i = 1, x \geq 0.
\end{array}
\end{equation}
Since $C$ is a positive semidefinite matrix, the function $\sqrt{x^T C x}$ is convex. When the probability level is set $p \geq 0.5$ (i.e., $\Phi^{-1} (1-p) \leq 0$), the objective function is concave so (\ref{kataoka}) turns out to be a convex programming problem.\\

\subsection{Currency rebalancing -- Conditional Value-at-Risk, i.e., minimization of risk}\label{stp2}

Optimization problems using Conditional Value-at-Risk (CVaR) have been researched and used in practice, so allow us to present the formulation of CVaR.
We recommend the readers to the literature, e.g.,  \cite{c1var},  \cite{risktom}, \cite{c2var} and the references therein.
In order to find the optimal portfolio using the classical CVaR, let $\xi$ denote the loss vector and $x$ the currency weights (the decision vector) for a portfolio of $n$ currencies.
Then $\xi^Tx$ means the loss of the currency holding, and the following model can be written:
\begin{equation}\label{cvar100}
\begin{array}{l}
\displaystyle \min_{x, a}  \left( a + \frac{1}{1-p} E\left([\xi^T x - a]_+\right) \right)\\
\begin{array}{ll}
\text{subject to} & \mu^T x \leq \mu_0\\
& \displaystyle \sum_{i=1}^n x_i \leq 1\\
& x \geq 0,
\end{array}
\end{array}
\end{equation}
where $\mu = E(\xi)$ and $\mu_0$ is some constant. Note that $a =\text{VaR}_p(X)$ at optimality and the optimal objective value is the smallest among all values of $E\left(\xi^T x \ | \ \xi^T x \geq \text{VaR}_p(\xi^Tx)\right)$ with $x$ in the feasible set of the constraints. Also note that we can write $\mu_0 = -R$ where $R$ means the minimum required return for the portfolio.

 It is well known that the following LP is a discrete version of (\ref{cvar100}) with $\mu_0 = -R$,
\begin{equation}\label{lp1}\begin{array}{l}
\displaystyle \min a + \frac{1}{K(1-p)} \sum_{k=1}^K u_k\\
\begin{array}{ll}
\text{subject to}&  a -x^T \mathbf{y_k}  + u_k \geq 0, \ k =1, \dots, K\\
&u_k \geq 0, \ k =1, \dots, K\\
& - x^T\mu \geq R\\
&\sum_{i=1}^n x_i \leq 1\\
&x_i \geq 0, i =  1, \dots, n,\\
\end{array}
\end{array}\end{equation}
where $\{ \mathbf{y_1}, \dots, \mathbf{y_K} \}$ denotes $K$ i.i.d. samples of the loss random vector $\xi \in R^n$. (\ref{lp1}) can equivalently be written in the following matrix notation:

\begin{equation}\label{lp2}\begin{array}{l}
\displaystyle \min a + \frac{1}{K(1-p)} \sum_{k=1}^K u_k\\
\text{subject to}\\[2ex]
\left( \begin{array}{ccc}
1_K & -{\bf Y}  & I_K \\
0    &  -\mu^T    & 0 \dots 0\\
  0 & -1^T_n & 0 \dots 0 
\end{array}     \right)  \left( \begin{array}{c}a \\ x \\ u \end{array} \right)
 \geq \left( \begin{array}{c}   0 \\ R \\ -1        \end{array} \right)\\
 u \geq 0, x \geq 0,
\end{array}\end{equation}
where $I_K$ means $K \times K$ identity matrix, $Y$ is a $K \times n$ matrix with samples of the loss random vector $\xi \in R^n$, and $1_K$ is a $K \times 1$ all ones vector.

\subsection{More general formulations}\label{stp3}
Let $x$ denote the vector of currency weights with its related cost vector $c$, and $\xi$ the random vector with an estimated distribution; the matrices $A$ and $T$ are present and future constraints, respectively. Then a more general stochastic programming model is formulated in the following way:

\begin{equation}\begin{array}{rll}
\min & c^T x  & \\
\text{subject to}  & A x  =  b,  x \geq 0 \\[1ex]
        & P(T x \geq \xi)  \geq  p, \\[1.2ex]
\end{array} \label{1.1.31}
\end{equation}
where $p$ is a fixed probability chosen by ourselves. In practice $p$ is near 1, for example we may choose $p$ as 0.8, 0.9, 0.95, 0.99, depending on our reliability requirement, i.e., in what proportion of the cases do we want the inequality $T x \geq \xi$ to be satisfied. Let $T_i$ denote the $i$th row of matrix $T$ and $\xi_i$ the $i$th component of random vector $\xi$. \\

There is a simplification possibility for the problem (\ref{1.1.31}). Instead of $P(Tx \geq \xi) \geq p$ we take 
$$E(\xi_i - T_i x \ | \ \xi_i - T_i x > 0) \leq d_i, i =1, \dots, r.$$
 If the function $g_i(z) = E(\xi_i - z \ | \ \xi_i - z > 0)$ is decreasing, then $g_i(T_i x) = E(\xi_i - T_i x \ | \ \xi_i - T_i x > 0) \leq d_i$ is equivalent to $T_i x \geq g_i^{-1}(d_i), i =1, \dots, r$ and the whole problem becomes:
\begin{equation}\begin{array}{rll}
\min & c^T x  & \\
\text{subject to}  & A x  =  b,  x \geq 0 \\[1ex]
        & T_i x \geq g_i^{-1}(d_i), i =1, \dots, r, \\[1.2ex]
\end{array} \label{logsim}
\end{equation}
which is an LP. \\

\begin{remark}[Related measures of violation]
In reliabiility theory and insurance problems $E(\xi - t \ | \ \xi - t > 0)$ is called ``Expected Residual Lifetime." It is natural that it is a decreasing function of $t$, but it is not always decreasing. (i.e., there are probability distributions for which it is not true, e.g., lognormal, Pareto if $t \geq 1$, etc.) If $\xi$ has a {\bf logconcave p.d.f.,} then $E(\xi - t \ | \ \xi - t > 0)$ is a decreasing function of $t$. \end{remark}

Note that the problem (\ref{1.1.31}) is a special case of the following general formulation:
\begin{equation}\label{1.1.2}
\begin{array}{l}
\displaystyle \min  h(x) \\
  \text{subject to}  \\
h_0(x) = P(g_1(x, \xi) \geq 0, \dots, g_r(x, \xi) \geq 0) \geq p_0\\
h_1(x) \geq p_1, \dots, h_m(x) \geq p_m,
\end{array}
\end{equation}
where $x$ is the decision vector, $\xi$ is a random vector, $h(x), h_1(x), \dots, h_m(x)$ are given functions, $0<p_0 \leq1, p_1, \dots, p_m$ are given numbers.\\

\begin{remark}[Convexity of the problem (\ref{1.1.2})]
Any logconcave function is quasi-concave, hence if $\xi \in R^r$ has a continuous distribution and logconcave density then $h_0(x)$ in problem (\ref{1.1.2}) is quasi-concave. Hence, $h_0(x)$ allows for the convex programming property. If the objective function $h$ is convex and we assume that $h_1, \dots, h_m$ are quasi-concave, then the problem is indeed convex.
\end{remark}

\begin{remark}
If the random vector $\xi$ has independent components $\xi_1, \dots, \xi_r$, then
\begin{equation}
P(Tx \geq \xi) = \displaystyle \prod_{i=1}^r P(T_i x \geq \xi_i) = \prod_{i=1}^r F_i(T_i x), 
\end{equation}
where $F_i$ is the c.d.f. of $\xi_i$ for $i=1, \dots, r$. The probabilistic constraint takes the form: 
\begin{equation}
\prod_{i=1}^r F_i(T_i x) \geq p.
\end{equation}
$x$ is feasible if this inequality, in addition to $Ax = b, x \geq 0$, is satisfied.
\end{remark}

Note that a variety of optimization applications can be found in \cite{stochastic} and the author's numerous excellent scientific articles.

\subsection{Valuation of boxdollar}

The value of boxdollar ($\mathcal{B}_0$) is initially pegged against a fiat currency ($\mathcal{M}_0$) held in the custodian vault. The deposited collateral will generate interest ($r$). Thus, according to the cost-of-carry theory, an arbitrage-free condition entails transfer of interest to the users of boxdollar, otherwise boxdollar will be discounted as much as to compensate for opportunity costs as follows:

\begin{equation}
\mathcal{B}_0=\mathcal{M}_0 \Longrightarrow \mathcal{B}_t=e^{rt}\mathcal{M}_0.
\end{equation}
However, as the users of boxdollar benefit from its convenience, the yield of convenience ($\delta$) can partially or utterly defray foregone interest rate or opportunity costs. In that case, the value of boxdollar in the future ($\mathcal{B}_t$) can maintain the initially pegged amount ($\mathcal{M}_0$) as such:

 \begin{equation}
\mathcal{B}_0=\mathcal{M}_0 \Longrightarrow \mathcal{B}_t=e^{(r-\delta)t}\mathcal{M}_0 \approx \mathcal{M}_0 \text{ if } \delta \approx r.
\end{equation}

\subsection{Boxdollar as a regional currency and a quasi-fiat money}

Boxdollar can be used as a regional currency with collateral on tax revenues. This way, the regional government can wield a multiplier effect for increased transactions of goods and services in the community. The key to succeeding a liquid, well-trusted local money lies in the trust among community members and the eventual guarantee of conversion to the fiat currency by the local government. In a well-coordinated game-theoretic setting the agreement among economic agents can create value, as the local currency can boost the regional economy as an ``outside'' money (\cite{goodhart}). A municipality with stable tax revenues can be a ideal place to adopt a cryptocurrency as its local currency. For example, Ithaca, NY--where ``Ithaca hours'' is circulated as a local money for shopping in local stores--can consider converting the conventional, note-based hours into a cryptocurrency for wider and frictionless usages. Those local governments with liquid, local cryptocurrencies can later enter into mutual agreements to render their currencies compatible. As mutual trust grows and strengthens, the locally and inter-municipally used cryptocurrencies can become a quasi-fiat money across the nation.

Suspicions as per whether a cryptocurrency can become a fiat money is not a new question as the history of finance has witnessed analogous events throughout the time. DLTs utilizing computers scattered around the globe to revise the same list of transactions appeal to finance's fundamental objective of ascertaining a transaction between two mutually-remote parties trustworthy to each other. The predecessors of DLTs have done similarly as banks in England established a clearinghouse in the late 18th Century to revamp the system of managing different ledgers of common transactions (\cite{kindleberger}). These banks issued new bank notes (``wagon-way through the air'') based on the collateral of gold coins (``earth-bound highway'', \cite{smith}). In another latter century, those who wanted to supply more money (``currency school'') clashed with their counterpart whose policy goal was to put the gushing bank notes on hold (``banking school''), which was a precursor of ``forking'' on cryptocurrencies. 

Through trials and errors, inflation control as the foremost policy target of central banking has confirmed a general equilibrium-theoretic prediction of the role of money as a num\'{e}raire of the modern economy such that goods and services can be valued at their absolute prices, not the relative prices or exchange ratios of a batter economy. Governance is the key to maintaining the banking system based on a currency. For a fiat money, it is the role of the central bank to grant undoubting trust on the money as a medium of payment and exchange, and a store of value. For a cryptocurrency, the governance of its ecosystem is the protocol of transaction confirmation. For our dual ledger-keeping boxchain, its ``2+2'' doubly-secure consensus protocol is the crux of governance or security which emboldens trust among the participants across the system of utility and incentives.

\subsection{The dual cryptocurrencies: boxcoin and boxdollar in the value chain}\label{description}

Boxdollar, a stablecoin will be used primarily as a medium of exchange. The participants will exchange their local fiat currencies to boxdollar based on a real time conversion rate to USD. This is quite simple but strong enough to obtain desirable level of trust and security -- one can put down a deposit with a US Dollar for every boxdollar issued (i.e., 1 to 1) so that the  boxdollar is asset backed and keeps its stability. 
There will be a unique and efficient digital wallet, called the boxpay-wallet. Using the boxpay-wallet, the market participants will see all previous transactions, conversion records, and current balances of boxdollar and boxcoin. If registered, fiat currencies in his or her bank accounts can be seen as well. The real time exchange rates of boxdollar to boxcoin, boxcoin to boxdollar and ones among all related fiat currencies are also presented on the boxpay-wallet.

The boxpay protocol consists of the most efficient dual currencies -- boxcoin and boxdollar on public and private networks, respectively. In what follows we list up some of the notable functions of boxdollar and boxcoin. Boxdollar has the following desirable features:
\begin{itemize}
\item Reliability and security
\item Medium of exchange 
\item Store of value
\item Diversification of currency holdings
\item Transferability (e.g., efficient transfer (cross-border) in P2P transactions, among merchants and customers, etc.)
\item Tracking all previous transactions
\end{itemize}

Boxcoin is mainly for making micropayment on public network (see Section \ref{box} for details). There are some specifics for the use of boxcoin, the key functions and benefits of boxcoin include the following:
\begin{itemize}
\item Rewards and Fees (as presented in Section \ref{incentivesystem})
\item Possibility of capital appreciation
\end{itemize}

Let's see from a standpoint of a buyer, say Alice. After she exchanged her local currency to  boxdollar (BXD) (based on a conversion rate to a USD), the calculated amount of BXDs will be stored in her boxpay-wallet and be ready to use. If Alice wants to use her local currency Chinese Yuan (CNY) to buy an item priced at BXD 500, first thing to do is to spend CNY 3,440 to receive BXD 500 into her boxpay-wallet. (Buying 1 US dollar for Chinese Yuan  requires CNY 6.88 using the conversion rates as of 1 PM (UTC -4), 8/14/18.) Note that it might be a good idea to put down more CNY to get more BXD if the USD appreciation is expected. (For example, the rate 1 USD = CNY 6.88 at the moment may later be changed to 1 USD $>$ CNY 6.88. It's been a while the USD gets more valuable compared to the other currencies.) She may have multiple items in her to-buy list even if she does not want to buy them at the moment.  After transaction completed the seller will be able to see the BXDs from the boxpay-wallet, ready to use in the e-marketplace (e.g., for shipping cost) or get an exchange for some local fiat currency (or multiple currencies) as needed. 


\begin{remark}[The boxpay wallet and smart insurance]

The boxpay wallet is now being developed,  and is a file (a simple database) of the digital keys, which are completely independent of the protocol. This will come with a smart insurance capability, systematically identifying claims to report.
\end{remark}





\section{Concluding remark}\label{bond}

We are in transition to a cashless society and, for quite some time, blockchain has been around as one of the most thriving technologies with a bright future ahead. Blockchain is a decentralized-distributed system which has revolutionized our perspective of the world. Many forms of digital currency have already been used in a variety of ways and places, e.g., in the online marketplaces and mobile banking systems, where it is now normal to use a phone number or an email address instead of a bank account. Although there are numerous benefits the digital currency has to offer, many cryptocurrencies have very low trading volumes as they failed to build their own ecosystem.

In this paper we introduce our new way of thinking on approach to DLT -- the dual ledger-keeping algorithm of the chain of antichains (boxchain). Our DLT uses the original transactions occurred in the DAG-based primal space. Using the dual ledger-keeping algorithm a chain of antichains is constructed in real time in the dual layer that synchronously reflects transactions in the DAG. This chain of antichains has resembles the structure of blockchain. Using such dual-layer framework, we take desirable aspects from both blockchain and a DAG-based distributed network system. We consequently arrive at a powerful consensus protocol which makes the final confirmation feasible with great efficiency. Our dual approach proposes two distinct cryptocurrencies: boxdollar and boxcoin. Our stablecoin boxdollar is  pegged to a fiat money, anticipating the use as a medium-of-exchange as well as a reliable store-of-value. Another cryptocurrency, called boxcoin, is a crucial component to keep up a purely distributed peer-to-peer network, with capability to run our unique incentive system at its core, indispensable to an effective DLT.

We presented both deterministic and stochastic aspects of the dual ledger-keeping. Illustrative and numerical examples are also presented. Our new algorithms were discovered by the use of discrete mathematics as well as probability models. This paper is focused on the mathematical and analytical foundations of a new digital ecosystem. We hope our new ideas for DLT will be beneficial to the readers and be helpful to improve decentralized-distributed network systems and their real-life applications.

\section{Acknowledgements}
The authors would like to express sincere gratitude to Professor Endre Boros for his kind critics, which inspired the authors to come up with a consensus algorithm. 
The authors would also like to thank the Box Protocol team for their comments and suggestions. 
 The first author dearly misses his academic father Professor Andr\'{a}s Pr\'{e}kopa (1929-2016).
 We also thank Daeje Chin and Joung Hwa Choi.

\bibliographystyle{ormsv080}      
\bibliography{MAVaR_refs_aor}   

\end{document}